\documentclass[journal]{IEEEtran}

\usepackage{color}
\usepackage{amssymb} 
\usepackage{mathtools}
\usepackage{footnote}
\usepackage{amsfonts}
\newcommand{\ve}[1]{{\bf #1}}
\newcommand{\mode}{\mbox{mod }}
\newcommand{\vol}{\mbox{Vol }}
\newcommand{\norm}[1]{\left|\left|#1\right|\right|}
\newtheorem{lemma}{Lemma}
\newtheorem{theorem}{Theorem}
\newtheorem{proposition}{Proposition}
\newtheorem{definition}{Definition}

\DeclarePairedDelimiter{\ceil}{\lceil}{\rceil}

\ifCLASSOPTIONcompsoc \usepackage[caption=false,font=normalsize,labelfon t=sf,textfont=sf]{subfig} \else \usepackage[caption=false,font=footnotesize]{subfi
g} \fi

\begin{document}

\sloppy
\IEEEoverridecommandlockouts
\title{Lattice Codes for Many-to-One  Interference Channels With and Without Cognitive Messages}

\author{Jingge Zhu and Michael Gastpar\IEEEmembership{, Member, IEEE}

\thanks{ Copyright (c) 2014 IEEE. Personal use of this material is permitted.  However, permission to use this material for any other purposes must be obtained from the IEEE by sending a request to pubs-permissions@ieee.org.}

\thanks{This work was supported in part by the European ERC Starting Grant 259530-ComCom. This paper was presented in part in IEEE International Symposium on Information Theory 2013, Istanbul, Turkey. }
\thanks{J. Zhu is with the School of Computer and Communication Sciences, Ecole Polytechnique F{\'e}d{\'e}rale de Lausanne (EPFL), Lausanne,
Switzerland (e-mail: jingge.zhu@epfl.ch).}

\thanks{M. Gastpar is with the School of Computer and Communication Sciences, Ecole Polytechnique F{\'e}d{\'e}rale de Lausanne (EPFL), Lausanne, Switzerland and the Department of Electrical Engineering and Computer Sciences, University of California, Berkeley, CA, USA (e-mail: michael.gastpar@epfl.ch).}

}



\maketitle

\begin{abstract}

A new achievable rate region is given for the Gaussian cognitive many-to-one  interference channel. The proposed novel coding scheme is based on the compute-and-forward approach with lattice codes. Using the idea of decoding sums of codewords, our scheme improves considerably upon the conventional coding schemes which treat interference as noise or decode messages simultaneously.   Our strategy also extends directly to the usual many-to-one interference channels without cognitive messages. Comparing to the usual compute-and-forward scheme where a fixed lattice is used for the code construction, the novel scheme employs scaled lattices and also encompasses key ingredients of the existing  schemes for the cognitive interference channel. With this new component, our scheme achieves a larger rate region in general. For some symmetric channel settings, new constant gap or capacity results are established, which are independent of the number of users in the system.

\end{abstract}

\section{Introduction}
Recently, with growing requests on high data rate and increasing numbers of intelligent communication devices, the concept of \textit{cognitive radio} has been intensively studied to boost  spectral efficiency. As one of its information-theoretic abstractions, a model of the cognitive radio channel of two users was proposed and analyzed in \cite{Devroye_etal_2006}, \cite{maric_capacity_2008}, \cite{jovicic_cognitive_2009}. In this model, the cognitive user
is assumed to know the message of the primary user non-causally
before transmissions take place.  The capacity region of this channel with additive white Gaussian noise is known for most of the
parameter region, see for example \cite{Rini_etal_2012} for an overview of the results.

In this work we extend this cognitive radio channel model to
include many cognitive users. We consider the simple many-to-one
interference scenario  with $K$ cognitive users illustrated in Figure \ref{fig:system}. The message $W_0$ (also called the \textit{cognitive message}) of the primary user is given to all other $K$ users, who could help the transmission of the primary user.  

Existing coding schemes for the cognitive interference channel exploit the usefulness of  cognitive messages. For the case $K=1,$ i.e., a single cognitive user,
the strategy consists in letting the cognitive user spend part of its resources to help the transmission of this message to the primary receiver. 
At the same time, this also appears as interference at the cognitive receiver. But dirty-paper coding can be used at the cognitive transmitter to cancel (part of) this interference.
A new challenge arises when there are many cognitive users. The primary user now benefits from the help of all cognitive users, but at the same time suffers from their collective interference because cognitive users are also transmitting their own messages. This inherent tension is more pronounced when the channels from cognitive transmitters to the primary receiver are strong. In the existing coding scheme, the interference from cognitive users is either decoded or treated as noise at the primary receiver.  As we will show later, direct extensions of these strategies to the many-to-one channel have significant shortcomings, especially when the interference is relatively strong.

The main contribution of this paper is a novel coding strategy for the cognitive interference network based on lattice codes. This scheme is based on the 
compute-and-forward approach (\cite{NazerGastpar_2011} \cite{ZhuGastpar_2014}). It deals with interference in a beneficial fashion, enabling some degree of  reconciliation between the competing factors mentioned above.  While most of the compute-and-forward work considers a fixed lattice to be used at each transmitter, the strategy developed here employs scaled lattices. In general it achieves larger rates than using fixed lattices and permits us to derive constant gap and capacity results.  We can also observe that the novel coding strategy encompasses several key ingredients of the existing coding schemes, such as rate splitting, dirty-paper coding and successive interference cancellation. The performance of the novel coding strategy is analyzed in detail.  We show our scheme outperforms conventional
coding schemes. The advantage is most notable in the case of strong interference from the cognitive users to the primary receiver. The proposed scheme applies naturally to the usual many-to-one interference channel, where the messages are not shared between users.   Applying the proposed scheme to a symmetric channel setting, we can show that under certain channel conditions, the novel coding strategy is near-optimal (in a constant-gap sense) or optimal regardless of the number of cognitive users.

The basic idea of the proposed scheme  is that instead of decoding its message directly, the primary
decoder first recovers enough linear combinations of messages
and then extracts its intended message.  Lattice codes are well suited for this purpose because their linear structure matches the additivity of the channels. More specifically, when two  codewords are superimposed additively, the resulting sum still lies in the lattice.  To give an intuitive explanation as to why this property is beneficial in the  interference channel, we note that as a general rule of thumb, the idea of interference alignment is needed in an interference network. However, using structured codes \textit{is} a form of interference alignment. When the interfering codewords are summed up linearly by the channels, the interference signal (more precisely, the sumset of the interfering codewords)  seen by the undesired receiver is much ``smaller'' when structured codes are used than when the codewords are chosen randomly. Hence the interference is ``aligned'' due to the linear structure of the codebook. This  property gives  powerful interference mitigation ability at the signal level.

Similar systems have been studied in the literature. For the case $K=2$, the system under consideration is studied in \cite{nagananda_multiuser_2013}.  A similar cognitive interference channel with  so-called cumulative message sharing is also studied in \cite{maamari_approximate_2014} where each cognitive user has messages of multiple users. We note that those existing results have not exploited the possibility of using structured codes in cognitive interference networks. The many-to-one channel without cognitive message is studied in \cite{Bresler_etal_2010}, where a similar idea of aligning interference based on lattice codes was used. We also point out that the method of compute-and-forward is versatile and  beneficial in many network scenarios. For example it has been used in \cite{Nam_etal_2010}, \cite{wilson_joint_2010} to study the Gaussian two-way relay  channel,  in \cite{ordentlich_approximate_2012} to study the $K$-user symmetric interference channel and in  \cite{Zhan_etal_2010}  to study the  multiple-antenna system.

The paper is organized as follows. Section \ref{sec:SystemModel}  introduces the system model and the problem statement. Section \ref{sec:ConventionalCoding} extends the known coding schemes from the two-user cognitive channel to the many-to-one cognitive channel. A novel coding scheme is proposed in Section \ref{sec:ComputeAndForward} where we also discuss its features in details. In Section \ref{sec:non-cognitive} we specialize our coding scheme to an interesting special case: the standard many-to-one interference channel without cognitive messages. We choose to present the cognitive channel first because it is more general and the results of the non-cognitive channel are absorbed in the former case.  

We use the notation $[a:b]$ to denote a set of increasing integers $\{a,a+1,\ldots,b\}$, $\log$ to denote $\log_2$ and $\log^+(x)$, $[x]^+$ to denote the function $\max\{\log(x),0\}, \max\{x,0\}$, respectively. We  use $\bar x$ for $1-x$ to lighten the notation at some places. We also adopt the convention that the sum $\sum_{i=m}^nx_i$ equals zero if $m>n$.

\section{System Model and problem statement}\label{sec:SystemModel}

We consider a multi-user channel consisting of $K+1$ transmitter-receiver pairs as shown in Figure \ref{fig:system}. The real-valued channel has the following vector representation:
\begin{IEEEeqnarray}{rCl}
\ve y_0 & = &\ve x_0+\sum_{k=1}^{K}b_k\ve x_k+\ve z_0,\\
\ve y_k & = & h_k\ve x_k+\ve z_k,\quad k\in[1:K],
\end{IEEEeqnarray}
where $\ve x_k$, $\ve y_k\in\mathbb R^n$ denote the channel input and output of the transmitter-receiver pair $k$, respectively. The noise $\ve z_k\in\mathbb R^n$ is assumed to be i.i.d.  Gaussian with zero mean and unit variance for each entry. Let $b_k\geq 0$ denote the channel gain from  Transmitter $k$ to the Receiver $0$ and $h_k$ denote the direct channel gain from Transmitter $k$ to its corresponding receiver for $k\in[1:K]$. We assume a unit channel gain for the first user without loss of generality. This system is sometimes referred to as the \textit{many-to-one interference channel} (or \textit{many-to-one channel} for simplicity), since only  Receiver $0$ experiences interference from other transmitters.

We assume that all  users have the same power constraint, i.e., the channel input $\ve x_k$ is subject to the power constraint 
\begin{align}
\mathbb E\{\norm{\ve x_k}^2\}\leq nP,\quad k\in[1:0].
\end{align}
Since channel gains are arbitrary, this assumption is without loss of generality. We also assume that all transmitters and receivers know their own channel coefficients; that is,   $b_k, h_k$ are known at  Transmitter $k$, $h_k$ is known at Receiver $k$, and $b_k, k\geq 1$ are known at Receiver $0$.

 \begin{figure}[htbp!]
   \centering
	\includegraphics[scale=0.48]{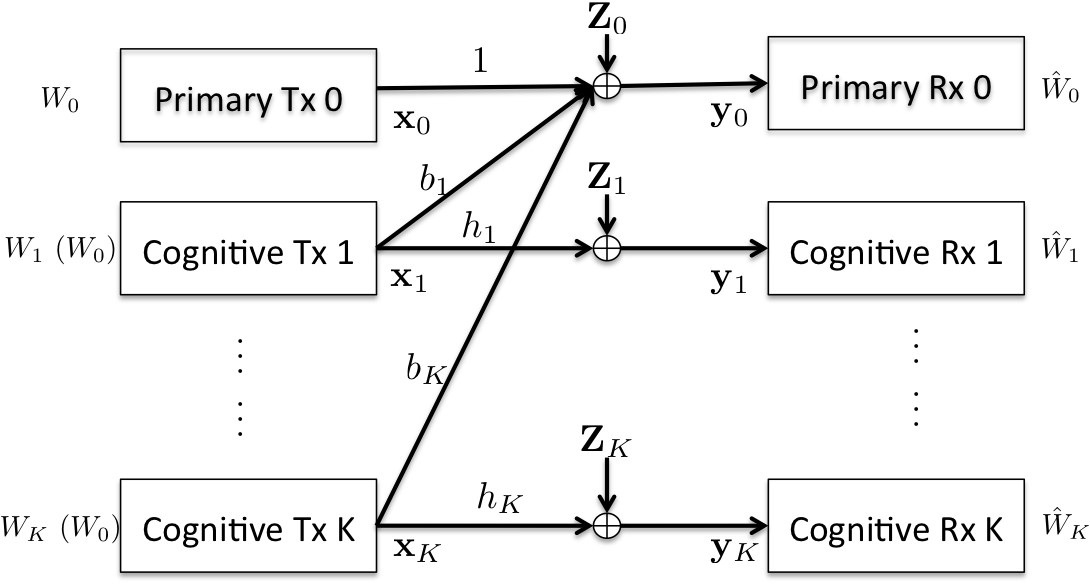} 
   \caption{A many-to-one interference channel. The message of the first user $W_0$ (called \textit{cognitive message}) may or may not be present at  other user's transmitter. }
   \label{fig:system}
 \end{figure}

Now we introduce two variants of this channel according to different message configurations.

\begin{definition}[Cognitive many-to-one channel]
User $0$ is called the primary user and User $k$ a cognitive user (for $k\geq 1$). Each user has a message $W_k$ from a set $\mathcal W_k$ to send to its corresponding receiver. Furthermore, all the cognitive users also have access to the primary user's message $W_0$ (also called cognitive message).
\end{definition}

\begin{definition}[Non-cognitive many-to-one channel]
Each user $k, k\in[0:K]$ has a message $W_k$ from a set $\mathcal W_k$ to send to its corresponding receiver. The messages are not shared among users.
\end{definition}

For the cognitive many-to-one channel, each transmitter has an encoder $\mathcal E_k: \mathcal W_k\rightarrow \mathbb R^n$ which maps the message to its channel input as 
\begin{IEEEeqnarray}{rCl}
\ve x_0&=&\mathcal E_k(W_0)\\
\ve x_k&=&\mathcal E_k(W_k,W_0),\quad k\in[1:K].
\end{IEEEeqnarray}

Each receiver has a decoder $\mathcal D_k:\mathbb R^n\rightarrow \mathcal W_k$ which estimates message  $\hat W_k$ from  $\ve y_k$ as 
\begin{align}
\hat W_k=\mathcal D_k(\ve y_k),\quad k\in[1:K].
\end{align}
The rate of each user is 
\begin{align}
R_k=\frac{1}{n}\log |\mathcal W_k|
\end{align}
under the average error probability requirement
\begin{align}
\mbox{Pr}\left(\bigcup_{k=0}^{K} \{\hat W_k\neq W_k\}\right)\rightarrow \epsilon
\end{align}
for any $\epsilon>0$.

For the non-cognitive many-to-one channel, the encoder takes the form
\begin{IEEEeqnarray}{rCl}
\ve x_k&=&\mathcal E_k(W_k),\quad k\in[0:K]
\end{IEEEeqnarray}
and other conditions are the same as in the cognitive channel.

As mentioned earlier, we find it convenient to first treat the  general model---the  cognitive many-to-one channel where  we derive a novel coding scheme which outperforms conventional strategies. We will show that the coding scheme for the cognitive channel can be extended straightforwardly to the non-cognitive channel, which also gives new results for this channel.

\section{Extensions of Conventional Coding Schemes}\label{sec:ConventionalCoding}
In this section we revisit existing coding schemes for the two-user cognitive interference channel and extend them to our cognitive many-to-one channel. The extensions are straightforward from the schemes proposed for the two-user cognitive channel in, for example, \cite{Devroye_etal_2006}, \cite{wu_capacity_2007} and \cite{Rini_etal_2012} .  Throughout the paper, many schemes can be parametrized by letting cognitive transmitters split their power.  For each cognitive user, we introduce a power splitting parameter $0\leq \lambda_k\leq 1$. For convenience, we also define the vector $\underline{\lambda}:=\{\lambda_1,\ldots,\lambda_K\}$.

%

In the first coding scheme, the cognitive users split the power and use part of it to transmit the message of the primary user. Luckily this part of the signal will not cause interference to the cognitive receiver since it can be completely canceled out using dirty-paper coding (DPC). We briefly describe the random coding argument for this coding scheme:
\begin{itemize}
\item \textbf{Primary encoder.} For each possible message $W_0$, User $0$ generates a codeword $\ve x_0$ with  i.i.d. entries  according to the Gaussian distribution $\mathcal N(0,P)$.
\item \textbf{Cognitive encoders.} User $k$ generates a sequence $\ve{\hat x}_k$ with i.i.d. entry according to the Gaussian distribution $\mathcal  N(0,\bar\lambda_kP)$ for any given $\lambda_k$ and form
\begin{align}
\ve u_k=h_k\ve{\hat x}_k+\gamma h_k \sqrt{\lambda_k}\ve x_0
\end{align}
with $\gamma=\bar\lambda_k h_k^2P/(1+\bar\lambda_k h_k^2 P)$, $k\geq 1$. The channel input is given by
\begin{IEEEeqnarray}{rCl}
\ve x_k=\sqrt{\lambda_k}\ve x_0+\ve{\hat x}_k,\quad k\in[1:K].
\end{IEEEeqnarray}
\item \textbf{Primary decoder.} Decoder $0$ decodes $\ve x_0$ from $\ve y_0$  using typicality decoding.
\item \textbf{Cognitive decoders.} Decoder $k$ ($k\geq 1$) decodes $\ve { u}_k$ from $\ve y_k$  using typicality decoding.
\end{itemize}
This coding scheme gives the following achievable rate region.
\begin{proposition}[DPC]
For the cognitive many-to-one channel, the above dirty paper coding scheme achieves the rate region:
\begin{IEEEeqnarray}{rCl}
R_0&\leq& \frac{1}{2}\log\left(1+\frac{(\sqrt{P}+\sum_{k\geq 1} b_k\sqrt{\lambda_k P})^2}{\sum_{k\geq 1} b_k^2\bar\lambda_k P+1}\right)\\
R_k&\leq& \frac{1}{2}\log\left(1+\bar\lambda_k h_k^2 P\right),\quad k\in[1:K]
\end{IEEEeqnarray}
for any power-splitting parameter $\underline\lambda$.
\label{prop:DPC}
\end{proposition}

It is worth noting that this scheme achieves the capacity in the two-user case ($K=1$) when $|b_1|\leq 1$, see \cite[Theorem 3.7]{wu_capacity_2007} for example.

Another coding scheme which performs well in the two-user case when $|b_1|>1$, is to let the primary user decode the  message of the cognitive user as well \cite{Rini_etal_2012}. We extend this scheme  by enabling \textit{simultaneous nonunique decoding} (SND) \cite[Ch. 6]{Elgamal_Kim_2011} at the primary decoder. SND improves the cognitive rates over uniquely decoding the messages $W_k, k\geq 1$ at primary decoder. We briefly describe the random coding argument for this coding scheme.
\begin{itemize}
\item \textbf{Primary encoder.} For each possible message $W_0$, User $0$ generates  a codewords $\ve x_0$ with  i.i.d. entries  according to the distribution $\mathcal N(0,P)$.
\item \textbf{Cognitive encoders.} Given the power splitting parameters $\lambda_k$,  user $k$ generates $\ve{\hat x}_k$ with i.i.d. entry according to the distribution $\mathcal  N(0,\bar\lambda_kP)$ for its message $W_k, k\geq 1$. The channel input is given by
\begin{IEEEeqnarray}{rCl}
\ve x_k=\sqrt{\lambda_k}\ve x_0+\ve{\hat x}_k
\end{IEEEeqnarray}
\item \textbf{Primary decoder.} Decoder $0$ simultaneously decodes  $\ve x_0, \ve{\hat x}_1,\ldots,\ve{\hat x}_K$ from $\ve y_1$  using typicality decoding. More precisely,  let $T^{(n)}(Y_0, X_0,\hat X_1\ldots, \hat X_K)$ denotes the set of $n$-length typical sequences (see, for example \cite[Ch. 2]{Elgamal_Kim_2011}) of the joint distribution $(\prod_{i=1}^KP_{\hat X_i})P_{X_0}P_{Y_0|X_0\ldots \hat X_K}$. The primary decoder decodes its message $\ve x_0$ such that
\begin{IEEEeqnarray}{rCl}
(\ve x_0, \hat{\ve x}_1,\ldots, \hat{\ve x}_K)\in T^{(n)}(Y_0, X_0,\hat X_1\ldots, \hat X_K)
\end{IEEEeqnarray}
for a unique $\ve x_0$ and \textit{some} $\hat{\ve x}_k, k\geq 1$.
\item \textbf{Cognitive decoders.} Decoder $k$ decodes $\ve{\hat x}_k$ from $\ve y_k$ for $k\geq 1$.
\end{itemize}

We have the following achievable rate region for the above coding scheme.
\begin{proposition}[SND at Rx $0$]
For the cognitive many-to-one channel, the above simultaneous nonunique decoding scheme achieves the rate region:
\begin{IEEEeqnarray*}{rCl}
R_0&\leq& \frac{1}{2}\log\left(1+\bigg(\sqrt{P}+\sum_{k\geq 1} b_k\sqrt{\lambda_k P}\bigg)^2\right)\\
R_0+\sum_{k\in\mathcal J} R_k&\leq& \frac{1}{2}\log\Bigg(1+\sum_{k\in\mathcal J} b_k^2\bar\lambda_k P \nonumber\\
&& + \bigg(\sqrt{P}+\sum_{k\geq 1} b_k\sqrt{\lambda_k P}\bigg)^2\Bigg)\\
R_k&\leq& \frac{1}{2}\log\left(1+\frac{\bar\lambda_k h_k^2P_k}{1+\lambda_k h_k^2P_k}\right)
\end{IEEEeqnarray*}
for any power-splitting parameter $\underline\lambda$ and every subset $\mathcal J\subseteq[1:K]$.
\label{prop:JointDecoding}
\end{proposition}

We point out that if instead of using simultaneous nonunique decoding at the primary decoder but require it to decode all messages of the cognitive users $W_k, k\geq 1$, we would have the extra constraints
\begin{IEEEeqnarray}{rCl}
\sum_{k\in\mathcal J} R_k&\leq& \frac{1}{2}\log\left(1+\sum_{k\in\mathcal J} b_k^2\bar\lambda_k P\right)
\end{IEEEeqnarray}
for every subset $\mathcal J\subseteq[1:K]$, which may further reduce the achievable rate region.

For the two-user case ($K=1$), the above scheme achieves the capacity when $|b_1|\geq \sqrt{1+P+P^2}+P$, see \cite[Theorem V.2]{Rini_etal_2012} for example.

We can further extend the above coding schemes by combining both dirty paper coding and SND at Rx $0$, as it is done in \cite[Theorem IV.1]{Rini_etal_2012}. However this  results in a very cumbersome rate expression in this system but gives little insight to the problem.  On the other hand, we will show in the sequel that our proposed scheme combines the ideas in the above two schemes in a unified framework.  

\section{A lattice codes based scheme for Cognitive Many-to-One  channels}\label{sec:ComputeAndForward}
In this section we provide a novel coding scheme for the  cognitive many-to-one channels based on a modified compute-and-forward scheme. The key idea of this approach is that instead of decoding the desired codeword directly at the primary receiver, it is more beneficial to first recover several integer combinations of the codewords and then solve for the desired message. We first briefly introduce the nested lattice codes used for this coding scheme and then describe how to adapt the compute-and-forward technique to our problem.

\subsection{Nested Lattice Codes}\label{subsec:LatticeCodes}
A lattice $\Lambda$ is a discrete subgroup of $\mathbb R^n$ with the property that if $\ve t_1, \ve t_2\in \Lambda$, then $\ve t_1+\ve t_2\in \Lambda$.  The details about lattice and lattice codes can be found, for example,  in \cite{Erez_etal_2005} \cite{ErezZamir_2004}. The lattice quantizer $Q_{\Lambda}: \mathbb R^n\rightarrow\Lambda$ is defined as as:
\begin{IEEEeqnarray}{rCl}
Q_{\Lambda}(\ve x)=\mbox{argmin}_{\ve t\in\Lambda}\norm{\ve t-\ve x}
\end{IEEEeqnarray}
The fundamental Voronoi region of a lattice $\Lambda$ is defined to be
\begin{IEEEeqnarray}{rCl}
\mathcal V:=\{\ve x\in\mathbb R^n:Q_{\Lambda}(\ve x)=\ve 0\}
\end{IEEEeqnarray}
The modulo operation gives the quantization error with respect to the lattice:
\begin{IEEEeqnarray}{rCl}
[\ve x]\mode\Lambda=\ve x-Q_{\Lambda}(\ve x)
\label{eq:mode}
\end{IEEEeqnarray}

Two lattices $\Lambda$ and $\Lambda'$ are said to be nested if $\Lambda'\subseteq\Lambda$. A nested lattice code $\mathcal C$ can be constructed  using the coarse $\Lambda'$ for \textit{shaping} and the fine lattice $\Lambda$ as codewords:
\begin{IEEEeqnarray}{rCl}
\mathcal C:=\{\ve t\in\mathbb R^n: \ve t\in\Lambda\cap\mathcal V'\}
\label{eq:nested_code}
\end{IEEEeqnarray}
where $\mathcal V'$ is the Voronoi region of $\Lambda'$. The \textit{second moment} of the lattice $\Lambda'$ per dimension is defined to be
\begin{IEEEeqnarray}{rCl}
\sigma^2(\Lambda')=\frac{1}{n\vol(\mathcal V')}\int_{\mathcal V'}\norm{\ve x}^2\mbox{d}\ve x
\label{eq:second_moment}
\end{IEEEeqnarray}
which is also the average power of code $\mathcal C$ defined in (\ref{eq:nested_code}) if the codewords $\ve t$ are uniformly distributed in $\mathcal V'$. 

The following two definitions are important for the lattice code construction considered here.
\begin{definition}[Good for AWGN channel]
Let $\ve z$ be a length-$n$ vector with i.i.d. Gaussian component $\mathcal N(0,\sigma^2_z)$, A sequence of $n$-dimensional lattices $\Lambda^{(n)}$ with its Voronoi region $\mathcal V^{(n)}$ is said to be good for AWGN channel if
\begin{IEEEeqnarray}{rCl}
\Pr(\ve z\notin\mathcal V^{(n)})\leq e^{-nE_p(\mu)}
\end{IEEEeqnarray}
where 
\begin{IEEEeqnarray}{rCl}
\mu=\frac{(\vol(\mathcal V^{(n)}))^{2/n}}{2\pi e\sigma_z^2}
\end{IEEEeqnarray}
is the volume-to-noise ratio and $E_p(\mu)$ is the Poltyrev exponent \cite{poltyrev_coding_1994} which is positive for $\mu>1$.
\label{def:AWGNGood}
\end{definition}
\begin{definition}[Good for quantization]
A sequence of $n$-dimensional lattices $\Lambda^{(n)}$ is said to be good for quantization if
\begin{IEEEeqnarray}{rCl}
\lim_{n\rightarrow\infty}\frac{\sigma^2(\Lambda^{(n)})}{(\vol(\mathcal V^{(n)}))^{2/n}} = \frac{1}{2\pi e}
\end{IEEEeqnarray}
with $\sigma(\Lambda^{(n)})^2$ denoting the second moment of the lattice $\Lambda^{(n)}$ defined in (\ref{eq:second_moment}). Notice the quantity on the LHS approachs the limit from above.
\label{def:QuantizationGood}
\end{definition}

Erez and Zamir \cite{ErezZamir_2004} have shown that there exist nested lattice codes where the fine lattice and the coarse lattice are both good for AWGN channel and good for quantization. Nam et al. \cite[Theorem 2]{Nam_etal_2011} extend the results to the case when there are multiple nested lattice codes.

Now we construct the nested lattice codes for our problem.  Let $\underline\beta:=\{\beta_0,\ldots,\beta_K\}$
denotes a set of  positive numbers.  For each user, we choose a lattice $\Lambda_k$ which is good for AWGN channel.  These $K+1$ fine lattices will form a  nested lattice chain \cite{Nam_etal_2011} according to a certain order which will be determined later. We let $\Lambda_c$ denote the coarsest lattice among them, i.e., $\Lambda_c\subseteq\Lambda_k$ for all $k\in[0:K]$.  As shown in \cite[Thm. 2]{Nam_etal_2011}, we can also find another $K+1$ simultaneously good nested lattices such that $\Lambda_k^s\subseteq\Lambda_c$ for all $k\in[0:K]$ whose second moments satisfy 
\begin{subequations}
\begin{align}
\sigma_0^2&:=\sigma^2(\Lambda_0^s)=\beta_0^2 P \\
\sigma_k^2&:=\sigma^2(\Lambda_k^s)=(1-\lambda_k)\beta_k^2P,\quad k\in[1:K]
\end{align}
\label{eq:sigma_sqr}
\end{subequations}
with given power-splitting parameters $\underline\lambda$.  Introducing the scaling coefficients $\underline{\beta}$ enables us to flexibly balance the rates of different users and utilize  the channel state information in a natural way. This point is made clear in the next section when we describe the coding scheme.  

The codebook for user $k$ is constructed as
\begin{align}
\mathcal C_k:=\{\ve t_k\in\mathbb R^n: \ve t_k\in\Lambda_k\cap\mathcal V_k^s\},\quad k\in[0:K]
\end{align}
where $\mathcal V_k^s$ denotes the Voronoi region of the \textit{shaping lattice} $\Lambda_k^s$  used to enforce the power constraints. With this lattice code, the message rate of user $k$ is also given by
\begin{IEEEeqnarray}{rCl}
R_k=\frac{1}{n}\log\frac{\vol(\mathcal V_k^s)}{\vol(\mathcal V_k)}
\end{IEEEeqnarray}
with $\mathcal V_k$ denoting the Voronoi region of the fine lattice $\Lambda_k$.

\subsection{Main Results}\label{sec:MainResult}
Equipped with the nested lattice codes constructed above, we are ready to specify the coding scheme. Each cognitive user splits its power and uses one part to help the primary receiver.  Messages $W_k\in\mathcal W_k$ of user $k$ are mapped surjectively to lattice points $\ve t_k\in\mathcal C_k$ for all $k$. 

Let $\underline{\gamma}=\{\gamma_1,\ldots,\gamma_K\}$ be $K$ real numbers to be determined later. Given  all messages $W_k$ and their corresponding lattice points $\ve t_k$, transmitters form
\begin{IEEEeqnarray}{rCl}
\ve x_0&=&\left[\frac{\ve t_0}{\beta_0}+\ve d_0\right]\mode\Lambda_0^s/\beta_0 \IEEEyessubnumber\\
\ve{\hat x}_k&=&\left[\frac{\ve t_k}{\beta_k}+\ve d_k-\frac{\gamma_k\ve x_0}{\beta_k}\right]\mode\Lambda_k^s/\beta_k,k\in[1:K] \IEEEyessubnumber
\label{eq:x0_xk}
\end{IEEEeqnarray}
where $\ve d_k$ (called \textit{dither}) is a random vector independent of $\ve t_k$ and uniformly distributed in $\mathcal V_k^s/\beta_k$.  It follows that $\ve{ x}_0$ is also uniformly distributed in $\mathcal V_0^s/\beta_0$ hence has average power $\beta_0^2P/\beta_0^2=P$ and is independent from $\ve t_0$ \cite[Lemma 1]{ErezZamir_2004}.  Similarly $\ve{\hat x}_k$ has average power $\bar\lambda_k P$ and is independent from $\ve t_k$ for all $k\geq 1$. 

Although $\ve x_0$ will act as interference at cognitive receivers, it is possible to cancel its effect at the receivers since it is known to cognitive transmitters. The dirty-paper coding idea in the previous section can also be implemented within the framework of lattice codes, see for example \cite{zamir_nested_2002}. The parameters $\underline{\gamma}$ are used to cancel  $\ve x_0$ partially or completely at the cognitive receivers.

The channel input for the primary transmitter is $\ve x_0$ defined above and the channel input for each cognitive transmitter is 
\begin{IEEEeqnarray}{rCl}
\ve x_k&=&\sqrt{\lambda_k}\ve x_0+\ve {\hat x}_k,\quad k\in[1:K].
\end{IEEEeqnarray}
Notice that  $\mathbb E\{\norm{\ve x_k}^2\}/n=\lambda_k P+\bar\lambda_kP=P$ hence power constraints are  satisfied for all cognitive users. 

We first give an informal description of the coding scheme and then present the main theorem. Let $\ve a:=[a_0,\ldots, a_K]\in\mathbb Z^{K+1}$ be a vector of integers.  We shall show that the integer sum of the lattice codewords $\sum_{k\geq 0}a_k\ve t_k$ can be decoded reliably at the primary user  for certain rates $R_k$.  After this, we continue decoding further integer sums with judiciously chosen coefficients and solve for the desired codeword using these sums at the end. An important observation  (also made in \cite{NazerGastpar_2011} and \cite{Nazer_2012}) is that the integer sums we have already decoded can be used to decode the subsequent integer sums. We now point out the new ingredients in our proposed scheme compared to the existing successive compute-and-forward schemes as in \cite{Nazer_2012} and  \cite{NazerGastpar_2011}. Firstly the scaling parameters introduced in (\ref{eq:sigma_sqr}) allow users to adjust there rates according to the channel gains and generally achieve larger rate regions. They will also be important for deriving constant gap and capacity results for the non-cognitive channel in Section \ref{sec:noncog}. Secondly as the cognitive message acts as interference at cognitive receivers,  using dirty-paper coding against the cognitive message in general improves the cognitive rates. But its implementation within successive compute-and-forward framework is not straightforward and requires careful treatment, as shown later in our analysis.

In general, let $L\in[1:K+1]$  be the total number of  integer sums\footnote{There is no need to decode more than $K+1$ sums since there are $K+1$ users in total.}  the primary user decodes and we represent the $L$ sets of coefficients in the following \textit{coefficient matrix}:
\begin{IEEEeqnarray}{rCl}
\ve A=\begin{pmatrix}
a_0(1) &a_1(1) &a_2(1) &\ldots &a_K(1)\\
\vdots &\vdots &\vdots &\vdots &\vdots\\
a_0(L) &a_1(L) &a_2(L) &\ldots &a_K(L)
\end{pmatrix},
\label{eq:matrix_representation}
\end{IEEEeqnarray}
where the $\ell$-th row $\ve a(\ell):=[a_0(\ell),\ldots,a_K(\ell)]$ represents the coefficients for the $\ell$-th integer sum $\sum_ka_k(\ell)\ve t_k$. We will show all $L$ integer sums can be decoded reliably if the rate of user $k$ satisfies
\begin{IEEEeqnarray}{rCl}
R_k&\leq &\min_{\ell} r_k(\ve a_{\ell|1:\ell-1},\underline{\lambda},\underline{\beta},\underline{\gamma})
\end{IEEEeqnarray}
with 
\begin{align}
r_k(\ve a_{\ell|1:\ell-1},\underline{\lambda},\underline{\beta},\underline{\gamma}):=\max_{\alpha_1,\ldots,\alpha_\ell\in\mathbb R}\frac{1}{2}\log^{+}\left(\frac{\sigma_k^2}{N_0(\ell)}\right).
\label{eq:Rk_SIC}
\end{align}
The notation $\ve a_{\ell|1:\ell-1}$ emphasizes the fact that when the primary decoder decodes the $\ell$-th sum with coefficients $\ve a(\ell)$, all previously decoded sums with coefficients $\ve a(1),\ldots,\ve a(\ell-1)$ are used.   In the expression above $\sigma_k^2$ is given in (\ref{eq:sigma_sqr}) and $N_0(\ell)$ is defined as
\begin{IEEEeqnarray}{rCl}
N_0(\ell)&&:=\alpha_\ell^2+\sum_{k\geq 1}\left(\alpha_\ell b_k-a_k(\ell)\beta_k-\sum_{j=1}^{\ell-1}\alpha_ja_k(j)\beta_k\right)^2\bar\lambda_k P \nonumber\\
&&+\left(\alpha_\ell b_0-a_0(\ell)\beta_0-\sum_{j=1}^{\ell-1}\alpha_ja_0(j)\beta_0-g(\ell)\right)^2P
\label{eq:N_0_l}
\end{IEEEeqnarray}
with 
\begin{IEEEeqnarray}{rCl}
b_0&:=&1+\sum_{k\geq 1}b_k\sqrt{\lambda_k}\label{eq:b_0}\\
g(\ell)&:=&\sum_{k\geq 1}\left(\sum_{j=1}^{\ell-1}\alpha_ja_k(j)+a_k(\ell)\right)\gamma_k. \label{eq:g_l}
\end{IEEEeqnarray}

For any matrix $\ve A\in\mathbb F_p^{L\times (K+1)}$,  let $\ve A'\in \mathbb F_p^{L\times K}$ denote the  matrix $\ve A$ without the first column. We define a set of matrices as
\begin{IEEEeqnarray}{rCl}
\mathcal A(L):=\{&&\ve A\in\mathbb F_p^{L\times (K+1)}:\mbox{rank}(\ve A)=m, \mbox{rank}(\ve A')=m-1 \nonumber\\
&& \mbox{for some integer }m, 1\leq m\leq L \}.
\label{eq:valid_coefficients}
\end{IEEEeqnarray}
We will show that if the coefficients matrix $\ve A$ of the $L$ integer sums is in this set, the desired codeword $\ve t_0$ can be reconstructed at the primary decoder.  For cognitive receivers, the decoding procedure is much simpler. They will decode the desired codewords directly using lattice decoding. 

Now we state the main theorem of this section formally and the proof will be presented in the next section.

\begin{theorem}
For any given set of power-splitting parameters $\underline\lambda$, positive numbers $\underline\beta$, $\underline{\gamma}$ and any coefficient matrix $\ve A\in\mathcal A(L)$ defined in (\ref{eq:valid_coefficients}) with $L\in[1:K+1]$,  define  $\mathcal L_k:=\{\ell\in[1:L]|a_k(\ell)\neq 0\}$.  If $r_k(\ve a_{\ell|1:\ell-1},\underline{\lambda},\underline{\beta},\underline{\gamma})>0$ for all $\ell\in\mathcal L_k$, $k\in[0:K]$,  then the following rate is achievable for the cognitive many-to-one interference channel
\begin{IEEEeqnarray}{rCl}
R_0&\leq& \min_{\ell\in\mathcal L_0}r_0(\ve a_{\ell|1:\ell-1},\underline{\lambda},\underline{\beta},\underline{\gamma})\IEEEyessubnumber \label{eq:R_0_cog}\\
R_k&\leq &\min\bigg\{\min_{\ell\in\mathcal L_k}r_k(\ve a_{\ell|1:\ell-1},\underline{\lambda},\underline{\beta},\underline{\gamma}), \nonumber\\
&&\quad\quad \quad \max_{\nu_k\in\mathbb R}\frac{1}{2}\log^+\frac{\sigma_k^2}{N_k(\gamma_k)} 
  \bigg\}\mbox{ for } k\geq 1.\IEEEyessubnumber\label{eq:R_k_cog}
\end{IEEEeqnarray}
The expressions $r_k(\ve a_{\ell|1:\ell-1},\underline{\lambda},\underline{\beta},\underline{\gamma})$ and  $\sigma_k^2$ are defined in  (\ref{eq:Rk_SIC}) and  (\ref{eq:sigma_sqr}) respectively,  and $N_k(\gamma_k)$ is defined as
\begin{align}
N_k(\gamma_k):=&\nu_k^2+(\nu_kh_k-\beta_k)^2\bar\lambda_k P\nonumber\\
&+(\nu_k\sqrt{\lambda_k} h_k-\gamma_k)^2P
\label{eq:N_k}
\end{align}
\label{thm:rate_cog}
\end{theorem}

Several comments are made  on the above theorem. We use $r_k(\ve a_{\ell|1:\ell-1})$ to denote $r_k(\ve a_{\ell|1:\ell-1},\underline{\lambda},\underline{\beta},\underline{\gamma})$ for brevity.
\begin{itemize}
\item In our coding scheme the primary user may decode more than one integer sums.  In general, decoding the $\ell$-th sum gives a constraint on $R_k$:
\begin{align}
R_k\leq r_k(\ve a_{\ell|1:\ell-1}).
\end{align}
However notice that if $a_k(\ell)=0$, i.e., the codeword $\ve t_k$ is not in the $\ell$-th sum, then $R_k$ does not have to be constrained by $r_k(\ve a_{\ell|1:\ell-1})$ since this decoding does not concern Tx $k$. This explains the minimization of $\ell$ over the set $\mathcal L_k$ in (\ref{eq:R_0_cog}) and (\ref{eq:R_k_cog}): the set $\mathcal L_k$ denotes all sums in which the codeword $\ve t_k$ participates and $R_k$ is determined by the minimum of $r_k(\ve a_{\ell|1:\ell-1})$ over $\ell$ in $\mathcal L_k$. 

\item Notice that $r_k(\ve a_{\ell|1:\ell-1})$ is not necessarily positive and a negative value means that the $\ell$-th sum cannot be decoded reliably. The whole decoding procedure will succeed only if all sums can be decoded successfully.  Hence in the theorem we require $r_k(\ve a_{\ell|1:\ell-1})>0$ for all $\ell\in\mathcal L_k$ to ensure that all sums can be decoded.

\item The primary user can choose which integer sums to decode, hence  can maximize the rate  over the number of integer sums $L$ and the coefficients matrix $\ve A$ in the set $\mathcal A(L)$, which gives the best rate as:
\begin{IEEEeqnarray*}{rCl}
R_k\leq  \max_{L\in[1:K+1]}\max_{ \ve A\in\mathcal A(L)} \min_{\ell\in\mathcal L_k}   r_k(\ve a_{\ell|1:\ell-1},\underline{\lambda},\underline{\beta},\underline{\gamma}).
\end{IEEEeqnarray*}
The optimal $\ve A$ is the same for all $k$. To see this, notice that the denominator inside the $\log$ of the expression $r_k(\ve a_{\ell|1:\ell-1})$ in (\ref{eq:Rk_SIC}) is the same for all $k$ and the numerator depends only on $k$ but does not involve the coefficient matrix $\ve A$, hence the maximizing $\ve A$ will be the same for all $k$.

\item In the expression of $r_k(\ve a_{\ell|1:\ell-1})$ in (\ref{eq:Rk_SIC}) we should optimize over $\ell$ parameters $\alpha_1,\ldots,\alpha_\ell$. The reason for involving these scaling factors is that there are two sources for the effective noise $N_0(\ell)$ at the lattice decoding stage, one is the non-integer channel gain and the other is the additive Gaussian noise in the channel. These scaling factors are used to balance these two effects and find the best trade-off between them, see \cite[Section III]{NazerGastpar_2011} for a detailed explanation.  The optimal $\alpha_\ell$ can be given explicitly but the expressions are very complicated  hence we will not state it here. We note that the expression $r_k(\ve a_1)$ with the optimized $\alpha_1$, $\beta_k=1$ and $\gamma_k=0$ is the computation rate  of compute-and-forward in \cite[Theorem 2]{NazerGastpar_2011}. 

\item As mentioned in  Section \ref{subsec:LatticeCodes}, the parameters $\underline{\beta}$ are used for controlling the rate of individual users.  Unlike the original compute-and-forward coding scheme in \cite{NazerGastpar_2011} where the transmitted signal $\ve x_k$ contains the lattice codeword $\ve t_k$  in the fine lattice $\Lambda$,  the transmitted signal here contains a scaled version of the lattice codeword $\ve t_k/\beta_k$. By choosing different $\beta_k$ for different user $k$ we can adjust the rate of the individual user and achieve a larger rate region in general. More information about this modified scheme can be found in \cite{ZhuGastpar_2014} where it is applied to other scenarios where the compute-and-forward technique is beneficial.

\item For the cognitive users, their rates are constrained both by their direct channel to the corresponding receiver, and by the decoding procedure at the primary user. The two terms in (\ref{eq:R_k_cog}) reflect these two constraints. The parameters $\underline{\gamma}$ are used to (partially) cancel the interference $\ve x_0$ at the cognitive receivers.  For example if we set $\gamma_k=\nu_k\sqrt{\lambda_k}h_k$, the cognitive receiver $k$ will not experience any interference caused by $\ve x_0$. However this affects the computation rate at the primary user in a non-trivial way through $r_k(\ve a_{\ell|1:\ell-1})$ (cf. Equations (\ref{eq:Rk_SIC}) and (\ref{eq:N_0_l})).
\end{itemize}

This proposed scheme can be viewed as an extension of the techniques used in the conventional schemes discussed in section \ref{sec:ConventionalCoding}.  First of all it includes the dirty-paper coding within the lattice codes framework and we can show the following lemma.
\begin{lemma}
The achievable rates in Proposition \ref{prop:DPC} can be recovered using Theorem \ref{thm:rate_noncog} by decoding one trivial sum with the coefficient $\ve a(1)=[1,0,\ldots, 0]$.
\end{lemma}
\begin{proof}
For  given power-splitting parameters $\underline{\lambda}$ we  decode only one trivial sum at the primary user by choosing $\ve a(1)$ such that $a_0(1)=1$ and $a_k(1)=0$ for $k\geq 1$, which is the same as decoding $\ve t_0$. First consider  decoding at the primary user. Using the expression (\ref{eq:Rk_SIC})  we have $R_k\leq  r_k(\ve a(1))=\frac{1}{2}\log(\sigma_k^2/N_0(1))$
with
$N_0(1)=\alpha_1^2\left(1+\sum_{k\geq 1}b_k^2\bar\lambda_k P\right)+(\alpha_1 b_0-\beta_0)^2P$
and $g(1)=0$ with this choice of $\ve a(1)$ for any $\underline{\gamma}$. After optimizing $\alpha_1$ we have
\begin{align}
R_0&\leq  \frac{1}{2}\log\left(1+\frac{b_0^2P}{1+\sum_{k\geq 1}b_k^2\bar\lambda_kP}\right).
\end{align}
Notice that  this decoding does not impose any constraint on $R_k$ for $k\geq 1$.

Now we consider the decoding process at the cognitive users. Choosing $\gamma_k=\nu_k\sqrt{\lambda_k}h_k$ in (\ref{eq:N_k}) will give
$N_k(\gamma_k)=\nu_k^2+(\nu_kh_k-\beta_k)^2\bar\lambda_k P$
and 
\begin{align}
\max_{\nu_k\in\mathbb R}\frac{1}{2}\log^+\frac{\sigma_k^2}{N_k(\gamma_k)}=\frac{1}{2}\log(1+h_k^2\bar\lambda_k P)
\end{align}
with the optimal $\nu_k^*=\frac{\beta_kh_k\bar\lambda_kP}{\bar\lambda_kh_k^2P+1}$.  This proves the claim. 
\end{proof}

The proposed scheme can also be viewed as an extension of  simultaneous nonunique decoding (Proposition \ref{prop:JointDecoding}). Indeed, as observed in \cite{bidokhti_nonunqie_2012}, SND can be replaced by either performing the usual joint (unique) decoding to decode all messages or treating interference as noise. The former case corresponds to decoding $K+1$ integer sums with a full rank  coefficient matrix and the latter case corresponds to decoding just one integer sum with the coefficients of  cognitive users' messages being zero. Obviously our scheme includes these two cases. As a generalization, the proposed scheme decodes just enough sums of codewords without decoding the individual messages. Unfortunately it is difficult to show analytically that the achievable rates in Proposition \ref{prop:JointDecoding} can be recovered using Theorem \ref{thm:rate_cog}, since it would require  the primary receiver to decode several non-trivial sums and the achievable rates are not analytically tractable for general channel gains. However the numerical examples in Section \ref{sec:Symmetric} will show that the proposed scheme generally performs better than the conventional schemes.

\subsection{On the Optimal Coefficient Matrix $\ve A$}\label{sec:optimal_A}
From Theorem \ref{thm:rate_cog} and its following comments we see that the main difficulty in evaluating the expression $r_k(\ve a_{\ell|1:\ell-1})$ in (\ref{eq:R_0_cog}) and (\ref{eq:R_k_cog}) is the maximization over all possible integer coefficient matrices in the set $\mathcal A(L)$.  This is an integer programming problem and is analytically intractable for a system with general channel gains $b_1,\ldots, b_K$. In this section we give an  explicit formulation of this problem and  an example of the choice of the coefficient matrix. 

The expression $r_k(\ve a_{\ell|1:\ell-1})$ in (\ref{eq:Rk_SIC}) is not directly amenable to analysis because finding the optimal solutions for the parameters $\{\alpha_\ell\}$ in (\ref{eq:N_0_l}) is prohibitively complex. Now we give an alternative formulation of the problem. We write $N_0(\ell)$ from Eq. (\ref{eq:N_0_l})  in the form of (\ref{eq:N_0_l_expand}).  It can be further rewritten compactly as
\begin{align}
N_0(\ell)=\alpha_{\ell}^2+\norm{\alpha_{\ell}\ve h-\tilde {\ve a}_{\ell}-\sum_{j=1}^{\ell-1}\alpha_j\tilde{\ve a}_j}^2P
\label{eq:N_0_l_compact}
\end{align}
where we define $\ve h, \tilde{\ve a}_j\in \mathbb R^K$  for $j\in[1:\ell]$ in (\ref{eq:h_a}).

\begin{figure*}[!htb]
\begin{IEEEeqnarray}{rCl}
N_0(\ell):=\alpha_\ell^2&&+\sum_{k\geq 1}\left(\alpha_\ell b_k\sqrt{\bar\lambda_k}-a_k(\ell)\beta_k\sqrt{\bar\lambda_k}-\sum_{j=1}^{\ell-1}\alpha_j
a_k(j)\beta_k\sqrt{\bar\lambda_k}\right)^2 P \nonumber\\
&&+\left(\alpha_\ell b_0-a_0(\ell)\beta_0-\sum_{k\geq 1}a_k(\ell)\gamma_k-\sum_{j=1}^{\ell-1}\alpha_j\left(a_0(j)\beta_0+\sum_{k\geq 1}a_k(j)\gamma_k\right)\right)^2P.
\label{eq:N_0_l_expand}
\end{IEEEeqnarray}
\begin{IEEEeqnarray}{rCl}
\ve h&=&\left[b_0, b_1\sqrt{\bar\lambda_1},\ldots, b_K\sqrt{\bar\lambda_k}\right]\nonumber\\
\tilde{\ve a}_j&=&\left[a_0(j)\beta_0+\sum_{k\geq 1}a_k(j)\gamma_k, a_1(j)\beta_1\sqrt{\bar\lambda_1},\ldots,a_K(j)\beta_K\sqrt{\bar\lambda_K} \right],\quad j\in[1:\ell]. \label{eq:h_a}
\end{IEEEeqnarray}
\noindent\rule{18cm}{0.8pt}
\end{figure*}


We will reformulate the above expression in such a way that the optimal parameters $\{\alpha_j\}$ have simple expressions and the optimization problem on $\ve A$ can be stated explicitly. This is shown in the following proposition.

\begin{proposition}
Given $\tilde{\ve a}_j, j\in[1:\ell-1]$ and $\ve h$ in (\ref{eq:h_a}), define
\begin{IEEEeqnarray}{rCl}
\ve u_{j}&=&\tilde{\ve a}_j-\sum_{i=1}^{j-1}\tilde{\ve a}_j|_{\ve u_i}, \quad j=1,\ldots \ell-1 \nonumber\\
\ve u_\ell&=&\ve h-\sum_{i=1}^{\ell-1}\ve h|_{\ve u_i}
\label{eq:u_j}
\end{IEEEeqnarray}
where $\ve x|_{\ve u_i}:=\frac{\ve x^T\ve u_i}{\norm{\ve u_i}^2}\ve u_i$ denotes the projection of a vector $\ve x$ on $\ve u_i$. The problem of finding the optimal coefficient matrix $\ve A$  maximizing $r_k(\ve a_{\ell|1:\ell-1})$ in Theorem \ref{thm:rate_cog} can be equivalently formulated as the following optimization problem
\begin{IEEEeqnarray}{rCl}
\min_{\substack{L\in[1:K+1]\\ \ve A\in\mathcal A(L)}} \max_{\ell\in\mathcal L_k} \norm{\ve B_{\ell}^{1/2}\ve a(\ell)}
\end{IEEEeqnarray}
where $\ve a(\ell)$ is the coefficient vector of the $\ell$-th integer sum. The set $\mathcal A(L)$ is defined in (\ref{eq:valid_coefficients}) and $\mathcal L_k:=\{\ell\in[1:L]|a_k(\ell)\neq 0\}$. The notation $\ve B_\ell^{1/2}$ denotes a  matrix satisfying\footnote{It is  shown that $N_0=P\ve a(\ell)^T\ve B_\ell\ve a(\ell)$ hence  $\ve B_{\ell}$ is positive semi-definite because $N_0\geq 0$. The guarantees the existence of $\ve B_\ell^{1/2}$.} $\ve B_\ell^{1/2}\ve B_\ell^{1/2}=\ve B_\ell$, where $\ve B_\ell$ is given by
\begin{IEEEeqnarray}{rCl}
\ve B_\ell&:=&\ve C\left(\ve I-\sum_{i=1}^{\ell-1}\frac{\ve u_j\ve u_j^T}{\norm{\ve u_j}^2}-\frac{(\ve u_\ell\ve u_\ell^T)P}{1+P\norm{\ve u_\ell}^2}\right)\ve C^T,
\label{eq:matrix_B}
\end{IEEEeqnarray}
and the matrix $\ve C$ is  defined as
\begin{align}
\ve C:=
\begin{pmatrix}
\beta_0 &0 &0 &\ldots &0\\
\gamma_1 &\beta_1\sqrt{\bar\lambda_1} &0 &\dots &0\\
\gamma_2 &0 &\beta_2\sqrt{\bar\lambda_2}&\ldots &0\\
\vdots &\vdots &\vdots &\vdots &\vdots\\
\gamma_K &0 &0 &\ldots &\beta_K\sqrt{\bar\lambda_K}
\end{pmatrix}.
\end{align}
\label{prop:optimization}
\end{proposition}
\begin{IEEEproof}
The proof is given in Appendix \ref{sec:proof_prop}.
\end{IEEEproof}

The above proposition makes the optimization of $\ve A$ explicit, although solving this problem is still a computationally expensive task. We should point out that this problem is related to the \textit{shortest vector problem}  (SVP) where one is to find the shortest non-zero vector in a lattice.  In particular let $\ve B\in\mathbb R^{K\times K}$ be a matrix whose columns constitute one set of basis vectors of the lattice, the SVP can be written as
\begin{align}
\min_{\ve a\in \mathbb Z^k,\ve a\neq \ve 0}\norm{\ve B\ve a}.
\end{align}
Our problem in Proposition \ref{prop:optimization} is more complicated than solving $L$ shortest vector problems. Because the $L$ matrices $\ve B_\ell^{1/2}$ are related through the optimal integer vectors $\ve a(\ell)$ in a nontrivial manner and the objective in our problem is to minimize the maximal vector length $\max_\ell \norm{\ve B_\ell^{1/2}\ve a(\ell)}$ of the $L$ lattices.  Furthermore the vectors $\ve a(1),\ldots, \ve a(\ell)$ should lie in the set $\mathcal A(L)$ and the number of sums $L$ is also an optimization variable.  A low complexity algorithm has been found to solve this instance of SVP  for the compute-and-forward problem in simple cases, see \cite{SahraeiGastpar_2014}.

Here we provide an example on the optimal number of sums we need to decode. Consider a  many-to-one channel with three cognitive users.  We assume $b_1=3.5$ and vary $b_2$ and $b_3$ in the range $[0,6]$.  We set  the direct channel gains $h_k=1$ and consider four different power constraints.  Now the goal is to maximize the sum rate 
\begin{align}
 \max_{\substack{L\in[1:4]\\ \ve A\in\mathcal A(L)}}\sum_{k=0}^4  \min_{\ell\in\mathcal L_k}  r_k(\ve a_{1:\ell-1},\underline{\lambda},\underline{\beta},\underline{\gamma})
\end{align}
with respect to $L\in[1:4]$, $\ve A\in\mathcal A(L)$ and $\underline{\beta}\in\mathbb R^4$. For simplicity we assume $\lambda_k=\gamma_k=0$ for $k\geq 1$. 
Here we search for all possible $\ve A$ and are interested in the optimal $L$: the optimal number of sums that need to be decoded.  
\begin{figure}[!h]
\hspace*{-0.31in}
\centering
 \includegraphics[scale=0.62]{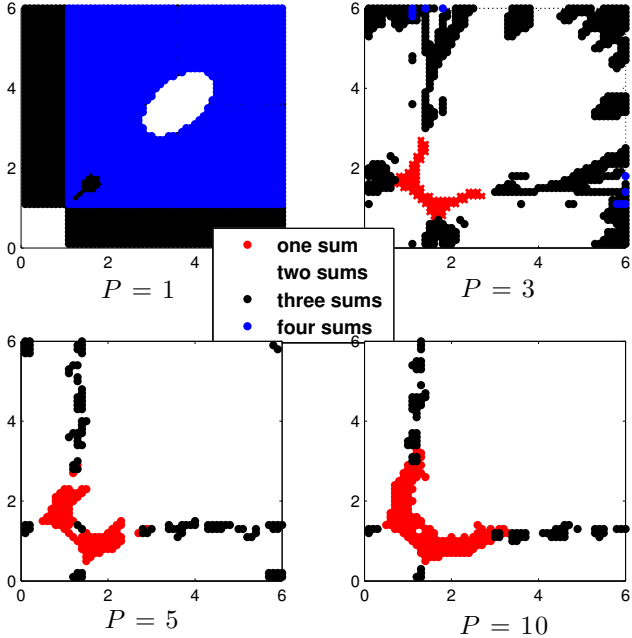}
  \caption{We consider a  many-to-one channel with three cognitive users and $b_1=3.5$. The horizontal and vertical axes are the range of $b_2$ and $b_3$, respectively. The objective is to maximize the sum rate. The red, white, black and blue  areas denote the region of different channel gains, in which the number of the best integer sums (the optimal $L$) is one, two,  three and four respectively. Here the patterns are shown for four different power constraints.}
\label{fig:NumEq}
\end{figure}

The four plots in Figure \ref{fig:NumEq} show the \textit{optimal number of integer sums} that the primary user will decode for different power constraints where $P$ equals $1, 3, 5$ or  $10$. The red area denotes the channel gains where the optimal $L$ equals $1$, meaning we need only decode one sum to optimize the sum rate, and so on. Notice that the sign of the channel coefficients $b_2, b_3$ will not change the optimization problem hence the patterns should be symmetric over both horizontal and vertical axes.  When power is small ($P=1$) we need to decode more than two sums in most channel conditions.  The patterns for $P$ equals $3, 5$ or $10$ look similar but otherwise  rather arbitrary--reflecting the complex nature of the solution to an integer programming problem. One observation from the plots is that  for $P$ relatively large, with most channel conditions we only need to decode two sums and we do not decode four sums, which is equivalent to solving for all messages. This confirms the point we made in the previous section: the proposed scheme generalizes the conventional scheme such as Proposition \ref{prop:JointDecoding}  to  decode just enough information for its purpose, but not more.

\subsection{Proof of  Theorem \ref{thm:rate_cog}}\label{sec:Proof_Thm_cog}
In this section we provide a detailed proof for Theorem \ref{thm:rate_cog}.  We also discuss the choice of the fine lattices $\Lambda_k$ introduced in \ref{subsec:LatticeCodes}.  The encoding procedure has been discussed in section \ref{sec:MainResult},  now we consider the decoding procedure at the primary user. The received signal $\ve y_0$ at the primary decoder is
\begin{IEEEeqnarray}{rCl}
\ve y_0&=&\ve x_0+\sum_{k\geq 1}b_k\ve x_k+\ve z_0\\
&=&(1+\sum_{k\geq 1}b_k\sqrt{\lambda_k})\ve x_0+\sum_{k\geq 1}b_k\ve{\hat x}_k+\ve z_0\\
&=&b_0 \ve{ x}_0+\sum_{k\geq 1}b_k\ve{\hat x}_k+\ve z_0
\label{eq:y_0}
\end{IEEEeqnarray}
where we define $b_0:=1+\sum_{k\geq 1}b_k\sqrt{\lambda_k}$. 

Given a set of integers $\ve a(1):=\{a_k(1)\in\mathbb Z,k\in[0:K]\}$ and some scalar $\alpha_1	\in\mathbb R$, the primary decoder  can form the following:
\begin{IEEEeqnarray*}{rCl}
\ve{\tilde y}_0^{(1)}&=&\alpha_1 \ve y_0-\sum_{k\geq 0}a_k(1)\beta_k\ve d_k \\
&=&(\alpha_1b_0-a_0(1)\beta_0)\ve x_0+  \sum_{k\geq 1}(\alpha_1 b_k -a_k(1)\beta_k)\ve{\hat  x}_k+\alpha_1\ve z_0\\
&&+\sum_{k\geq 1}a_k(1)\beta_k\ve{\hat x}_k+a_0(1)\beta_0\ve x_0-\sum_{k\geq 0}a_k(1)\beta_k\ve d_k.
\end{IEEEeqnarray*}
Rewrite the last three terms in the above expression as
\begin{IEEEeqnarray}{rCl}
&&\sum_{k\geq 1}a_k(1)\beta_k\ve{\hat x}_k+a_0(1)\beta_0\ve x_0-\sum_{k\geq 0}a_k(1)\beta_k\ve d_k\nonumber\\
&\stackrel{(b)}{=}&\sum_{k\geq 1}a_k(1)\left(\beta_k(\frac{\ve t_k}{\beta_k}-\frac{\gamma_k\ve x_0}{\beta_k})-\beta_kQ_{\frac{\Lambda_k^s}{\beta_k}}(\frac{\ve t_k}{\beta_k}+\ve d_k-\frac{\gamma_k\ve x_0}{\beta_k})\right) \nonumber\\
&&+a_0(1)\left(\beta_0\ve t_0-\beta_0Q_{\frac{\Lambda_0^s}{\beta_0}}(\frac{\ve t_0}{\beta_0}+\ve d_0) \right)\nonumber\\
&\stackrel{(c)}{=}&-\sum_{k\geq 1}a_k(1)\gamma_k\ve x_0+a_0(1)(\ve t_0-Q_{\Lambda_0^s}(\ve t_0+\beta_0\ve d_0))
 \nonumber\\
&&+ \sum_{k\geq 1}a_k(1)\left(\ve t_k-Q_{\Lambda_k^s}(\ve t_k+\beta_k\ve d_k-\gamma_k\ve x_0)\right) \nonumber\\
&\stackrel{(d)}{=}&-\sum_{k\geq 1}a_k(1)\gamma_k\ve x_0+\sum_{k\geq 0}a_k(1)\ve{\tilde t}_k. \label{eq:tmp1}
\end{IEEEeqnarray}
In step $(b)$ we used the definition of the signals $\ve x_0$ and $\ve{\hat x}_k$ from Eqn. (\ref{eq:x0_xk}). Step $(c)$ uses the identity  $Q_{\Lambda}(\beta \ve x)=\beta Q_{\frac{\Lambda}{\beta}}(\ve x)$ for any real number $\beta\neq 0$.  In step $(d)$ we define $\ve {\tilde t}_k$ for user $k$ as
\begin{IEEEeqnarray}{rCl}
\ve {\tilde t}_0&:=&\ve t_0-Q_{\Lambda_0^s}(\ve t_0+\beta_k\ve d_0)\\
\ve {\tilde t}_k&:=&\ve t_k-Q_{\Lambda_k^s}(\ve t_k+\beta_k\ve d_k-\gamma_k\ve x_0)\quad k\in[1:K].
\label{eq:tilde_t_k}
\end{IEEEeqnarray}
Define $g(1):=\sum_{k\geq 1}a_k(1)\gamma_k$ and substitute the expression (\ref{eq:tmp1}) into $\tilde{\ve y}_0^{(1)}$ to get
\begin{IEEEeqnarray}{rCl}
\tilde{\ve y}_0^{(1)}&=&\left(\alpha_1b_0-a_0(1)\beta_0-g(1)\right)\ve x_0+\sum_{k\geq 1}(\alpha_1 b_k -a_k(1)\beta_k)\ve{\hat  x}_k \nonumber\\
&&+\alpha_1\ve z_0+\sum_{k\geq 0}a_k(1)\tilde{\ve t}_k\nonumber\\
&=&\tilde{\ve z}_0(1)+\sum_{k\geq 0}a_k(1)\tilde{\ve t}_k
\label{eq:y0_tilde}
\end{IEEEeqnarray}
where we define the equivalent noise $\ve{\tilde z}_0(1)$  at the primary receiver as:
\begin{IEEEeqnarray}{rCl}
\ve {\tilde z}_0(1)&:=&\alpha_1 \ve z_0+(\alpha_1 b_0-a_0(1)\beta_0-g(1))\ve x_0 \nonumber\\
&&+\sum_{k\geq 1} (\alpha_1 b_k-a_k(1)\beta_k)\ve {\hat x}_k
\label{eq:z_0_1_tilde}
\end{IEEEeqnarray}
where $b_0:=1+\sum_{k\geq 1}b_k\sqrt{\lambda_k}$.

Notice that  we have $\ve{\tilde t}_k\in\Lambda_k$ since $\ve t_k\in\Lambda_k$  and  $\Lambda_k^s\subseteq\Lambda_c$ due to the lattice code construction (recall that $\Lambda_c$ denotes the coarsest lattice among all $\Lambda_k$ for $k\in[0:K]$). Furthermore because all $\Lambda_k$ are chosen to form a nested lattice chain,  the integer combination $\sum_{k\geq 0}a_k(1)\ve{\tilde t}_k$ also belongs to the finest lattice  among all $\Lambda_k$ with $a_k(1)\neq 0$. We denote this finest lattice as $\Lambda_f$, i.e., $\Lambda_k\subseteq\Lambda_f$ for all $k\in[0:K]$ satisfying $a_k(1)\neq 0$.  Furthermore, the equivalent noise $\ve{\tilde z}_0(1)$ is independent of the signal $\sum_{k\geq 0}a_k(1)\ve{\tilde t}_k$ thanks to the dithers $\ve d_k$.

The primary decoder performs \textit{lattice decoding} to decode the integer sum $\sum_{k\geq 0}a_k(1)\ve{\tilde t_k}$ by quantizing $\ve{\tilde y}_0^{(1)}$ to its nearest neighbor in $\Lambda_f$. A decoding error occurs when $\ve{\tilde y}_0^{(1)}$ falls outside the Voronoi region around the lattice point $\sum_{k\geq 0}a_k(1)\ve{\tilde t_k}$. The probability of this event is equal to the probability that the equivalent noise $\ve{\tilde z}_0(1)$ leaves the Voronoi region of the finest lattice, i.e., $\mbox{Pr}(\ve{\tilde z}_0(1)\notin\mathcal V_f)$ where $\mathcal V_f$ denotes the Voronoi region of $\Lambda_f$. The same as in the proof of \cite[Theorem 5]{NazerGastpar_2011}, the probability $\mbox{Pr}(\ve{\tilde z}_0(1)\notin\mathcal V_f)$ goes to zero if the probability $\mbox{Pr}(\ve z_0^*(1)\notin\mathcal V_f)$ goes to zero where $\ve z_0^*(1)$ is a zero-mean Gaussian vector with i.i.d entries whose variance equals the variance of the noise $\ve{\tilde z}_0(1)$:
\begin{IEEEeqnarray*}{rCl}
N_0(1)&=&\alpha_1^2+(\alpha_1 b_0-a_0(1)\beta_0-g(1))^2P \\
&&+\sum_{k\geq 1}(\alpha_1 b_k-a_k(1)\beta_k)^2\bar\lambda_kP.
\end{IEEEeqnarray*}

By the AWGN goodness property (Definition \ref{def:AWGNGood}) of $\Lambda_f$, the probability $\mbox{Pr}(\ve z_0^*(1)\notin\mathcal V_f)$ goes to zero exponentially if
\begin{IEEEeqnarray}{rCl}
\frac{(\vol(\mathcal V_f))^{2/n}}{N_0(1)}>2\pi e.
\label{eq:decoding_fine_lattice}
\end{IEEEeqnarray}
Since $\Lambda_f$ is the finest lattice in the nested lattice chain formed by $\Lambda_k, k\in[0:K]$ satisfying $a_k(1)\neq 0$, namely
\begin{align*}
\vol(\mathcal V_f)=\min_{k\in[0:K], a_k(1)\neq 0} \vol(\mathcal V_k),
\end{align*}
the inequality in (\ref{eq:decoding_fine_lattice}) holds if it holds that
\begin{IEEEeqnarray}{rCl}
\frac{(\vol(\mathcal V_k))^{2/n}}{N_0(1)}>2\pi e.
\end{IEEEeqnarray}
for all $k\in[0:K]$ satisfying $a_k(1)\neq 0$.  Hence using the rate expression 
\begin{IEEEeqnarray}{rCl}
R_k=\frac{1}{n}\log\frac{\vol(\mathcal V_k^s)}{\vol{\mathcal (\mathcal V_k)}} 
\label{eq:V_k_1sum}
\end{IEEEeqnarray}
we see the error probability goes to zero, or equivalently (\ref{eq:decoding_fine_lattice}) holds,  if
\begin{IEEEeqnarray}{rCl}
2^{2R_k}\leq \frac{(\vol(\mathcal V_k^s))^{2/n}}{2\pi e N_0(1)}
\end{IEEEeqnarray}
for all $k\in[0:K]$ satisfying $a_k(1)\neq 0$.  For Tx $k$ with $a_k(1)=0$, decoding this integer sum will not impose any constraint on the rate $R_k$.

Recalling the fact that $\Lambda_k^s$ is good for quantization (Definition \ref{def:QuantizationGood}), we have 
\begin{align}
\frac{\sigma_k^2}{(\vol(\mathcal V_k^s))^{2/n}}<\frac{(1+\delta)}{2\pi e}
\end{align}
for any $\delta>0$. We conclude that  lattice decoding will be successful if 
\begin{IEEEeqnarray}{rCl}
R_k< r_k(\ve a_1,\underline{\lambda},\underline{\beta},\underline{\gamma}):=\frac{1}{2}\log\frac{\sigma_k^2}{N_0(1)}-\frac{1}{2}\log(1+\delta)
\label{eq:R_k_OneEq}
\end{IEEEeqnarray}
that is
\begin{subequations}
\begin{align}
R_0&<\frac{1}{2}\log^+\left(\frac{\beta_0^2P}{\alpha_1^2+P\norm{\alpha_1\ve h-\ve{\tilde a}}^2}\right)\\
R_k&<\frac{1}{2}\log^+\left(\frac{(1-\lambda_k)\beta_k^2P}{\alpha_1^2+P\norm{\alpha_1\ve h-\ve{\tilde a}}^2}\right)\quad k\in[1:K]
\end{align}
\end{subequations}
if we choose $\delta$ arbitrarily small and define
\begin{IEEEeqnarray*}{rCl}
\ve h&:=&[b_0,b_1\sqrt{\bar\lambda_1},\ldots,b_K\sqrt{\bar\lambda_K}]\\
\ve{\tilde a}&:=&[a_0(1)\beta_0+g(1),a_1(1)\beta_1\sqrt{\bar\lambda_1},\ldots,a_K(1)\beta_K\sqrt{\bar\lambda_K}].
\end{IEEEeqnarray*}
Notice we can optimize over $\alpha_1$ to maximize the  above rates.

At this point, the primary user has  successfully decoded one integer sum of the lattice points $\sum_{k\geq 0}a_k\ve{\tilde t}_k$.  As mentioned earlier, we may  continue decoding other integer sums with the help of this sum. The method of performing \textit{successive compute-and-forward} in \cite{Nazer_2012} is to first recover a linear combination of all transmitted signals $\tilde{\ve x}_k$ from the decoded integer sum and use it for subsequent decoding. Here we are not able to do this because the cognitive channel input $\hat{\ve x}_k$ contains $\ve x_0$ which is not known at Receiver $0$.   In order to proceed, we use the observation that if $\sum_{k\geq 0}a_k\ve{\tilde t}_k$ can be decoded reliably, then we know the equivalent noise $\tilde{\ve z}_0(1)$ and can use it for the subsequent decoding.

In general assume the primary user has decoded $\ell-1$ integer sums $\sum_ka_k(j)\ve t_k$, $j\in[1:\ell-1], \ell\geq 2$ with positive rates, and about to decode another integer sum with coefficients $\ve a(\ell)$. 
We show in  Appendix \ref{app:SIC}  that with the previously known $\tilde{\ve z}_0(\ell-1)$ for $\ell\geq 2$, the primary decoder can form
\begin{IEEEeqnarray}{rCl}
\ve {\tilde  y}_0^{(\ell)}
&=&\ve{\tilde z}_0(\ell)+\sum_{k\geq 0}a_k(\ell)\ve{\tilde t}_k
\end{IEEEeqnarray}
with the equivalent noise $\ve{\tilde z}_0(\ell)$
\begin{IEEEeqnarray}{rCl}
\ve{\tilde  z}_0(\ell)&&:=\alpha_\ell\ve z_0+\sum_{k\geq 1}\left( \alpha_\ell b_k-a_k(\ell)\beta_k-\sum_{j=1}^{\ell-1}\alpha_j a_k(j)\beta_k\right)\hat{\ve x}_k\nonumber \\
&&+\left(\alpha_\ell b_0-a_0(\ell)\beta_0-\sum_{j=1}^{\ell-1}\alpha_ja_0(j)\beta_0-g(\ell)\right)\ve x_0
\label{eq:z_0_l_tilde}
\end{IEEEeqnarray}
where $g(\ell)$ is defined in (\ref{eq:g_l}) and the scaling factors $\alpha_1,\ldots,\alpha_\ell$ are to be optimized.

In the same vein as  we derived  (\ref{eq:R_k_OneEq}), using $\ve{\tilde y}_0^{(l)}$ we can decode the  integer sums of the lattice codewords $\sum_{k\geq 0}a_k(\ell)\ve {\tilde t_0}$ reliably  using lattice decoding if the fine lattice satisfy
\begin{IEEEeqnarray}{rCl}
\frac{(\vol(\mathcal V_k))^{2/n}}{N_0(\ell)}>2\pi e
\label{eq:V_k_lsum}
\end{IEEEeqnarray}
for $k$ satisfying $a_k(\ell)\neq 0$ and we use $N_0(\ell)$ to denote the  variance of the equivalent noise $\ve{\tilde z}_0(\ell)$  per dimension given in (\ref{eq:N_0_l}). Equivalently we require the rate $R_k$ to be smaller than
\begin{IEEEeqnarray}{rCl}
r_k(\ve a_{\ell|1:\ell-1},\underline{\lambda},\underline{\beta},\underline{\gamma}):=\max_{\alpha_1,\ldots,\alpha_\ell\in\mathbb R}\frac{1}{2}\log^{+}\left(\frac{\sigma_k^2}{N_0(\ell)}\right)
\end{IEEEeqnarray}
where $\sigma_k^2$ is given in (\ref{eq:sigma_sqr}). Thus we arrive at the same expression   in (\ref{eq:Rk_SIC}) as claimed.

Recalling the definition of the set $\mathcal A(L)$ in (\ref{eq:valid_coefficients}), we now show that if the coefficient matrix $A$ is in this set, the term $\ve{\tilde t}_0$ can be solved using the $L$ integer sums with coefficients $\ve a(1),\ldots,\ve a(L)$.

For the case $\mbox{rank}(\ve A)=K+1$ the statement is trivial. For the case $\mbox{rank}(\ve A)=m\leq L<K+1$, we know that by performing Gaussian elimination on $\ve A'\in \mathbb Z^{L\times K}$ with rank $m-1$,  we obtain a matrix whose last $L-m+1$ rows are zeros. Notice that $\ve A\in\mathbb Z^{L\times K+1}$ is a matrix formed by adding one more column in front of $\ve A'$. So if we perform  exactly the same Gaussian elimination procedure on the matrix $\ve A$,  there must be at least one row in the last $L-m+1$ row whose first entry is non-zero, since $\text{rank}(\ve A)=\text{rank}(\ve A')+1$. This row will give the value of $\ve {\tilde t}_0$. Finally the true codeword $\ve t_0$ can be recovered as 
\begin{align}
\ve t_0=[\tilde{\ve t}_0]\mode\Lambda_0^s.
\end{align}

Now we consider the decoding procedure at  cognitive receivers, for whom it is just a point-to-point transmission problem over Gaussian channel using lattice codes. The cognitive user $k$  processes its received signal for some $\nu_k$ as
\begin{IEEEeqnarray*}{rCl}
\ve {\tilde y}_k&=&\nu_k\ve y_k-\beta_k\ve d_k\\
&=&\nu_k(\ve z_k+\sqrt{\lambda_k}h_k\ve x_0)+(\nu_kh_k-\beta_k)\ve{\hat x}_k+\beta_k\ve{\hat x}_k-\beta_k\ve d_k\\
&\stackrel{}{=}&\nu_k(\ve z_k+\sqrt{\lambda_k}h_k\ve x_0)+(\nu_kh_k-\beta_k)\ve{\hat x}_k-\beta_k\ve d_k\\
&&+Q_{\Lambda_k^s}(\ve t_k+\beta_k\ve d_k-\gamma_k\ve x_0)+\beta_k(\frac{\ve t_k}{\beta_k}+\ve d_k-\frac{\gamma_k}{\beta_k}\ve x_0)\\
&=&\ve{\tilde z}_k+\tilde{\ve t}_k.
\end{IEEEeqnarray*}
In the last step we define the equivalent noise as
\begin{IEEEeqnarray}{rCl}
\ve{\tilde z}_k&:=&\nu_k\ve z_k+(\nu_kh_k-\beta_k)\ve{\hat x}_k+(\nu_k\sqrt{\lambda_k}h_k-\gamma_k)\ve x_0
\end{IEEEeqnarray}
and $\tilde{\ve t}_k$ as in (\ref{eq:tilde_t_k}).

Using the same argument as before,  we can show that the codeword $\tilde{\ve t}_k$ can be decoded reliably using lattice decoding if
\begin{IEEEeqnarray}{rCl}
\frac{(\vol(\mathcal V_k))^{2/n}}{N_k(\gamma_k)}>2\pi e
\label{eq:V_k_cog}
\end{IEEEeqnarray}
for all $k\geq 1$ where $N_k(\gamma)$ is the variance of the equivalent noise $\tilde{\ve z}_k$ per dimension given in (\ref{eq:N_k}).  Equivalently the cognitive rate $R_k$ should satisfy
\begin{IEEEeqnarray}{rCl}
R_k&<&\max_{\nu_k}\frac{1}{2}\log\frac{\sigma_k^2}{N_k(\gamma_k)}.
\end{IEEEeqnarray}
Similarly we can obtain $\ve t_k$ from $\tilde{\ve t}_k$ as $\ve t_k=[\tilde{\ve t}_k]\mode\Lambda_k^s$. This completes the proof of Theorem \ref{thm:rate_cog}.

We also determined how to choose the fine lattice $\Lambda_k$. Summarizing the requirements in (\ref{eq:V_k_cog}) and (\ref{eq:V_k_lsum}) on $\Lambda_k$ for  successful decoding, the fine lattice $\Lambda_0$ of the primary user  satisfies
\begin{IEEEeqnarray}{rCl}
(\vol(\mathcal V_0))^{2/n}>2\pi e N_0(\ell)
\end{IEEEeqnarray}
for all $\ell$ where $a_0(\ell)\neq 0$ and the fine lattice $\Lambda_k$ of the cognitive user $k$, $k\in[1:K]$,  satisfies
\begin{IEEEeqnarray}{rCl}
(\vol(\mathcal V_k))^{2/n}>\max\{2\pi e N_0(\ell), 2\pi e N_k(\gamma_k)\}
\end{IEEEeqnarray}
for all $\ell$ where $a_k(\ell)\neq 0$. As mentioned in Section \ref{subsec:LatticeCodes}, the fine lattices $\Lambda_k$ are chosen to form a nested lattice chain. Now the order of this chain can be determined by the volumes of $\mathcal V_k$ given above.

\subsection{Symmetric Cognitive Many-to-One Channels}\label{sec:Symmetric}
As we have seen in  Section \ref{sec:optimal_A}, it is in general difficult to describe the optimal coefficient matrix $A$. However we can give a partial answer to this question if we focus on one simple class of  many-to-one channels. In this section we consider a symmetric system with $b_k=b$ and $h_k=h$ for all $k\geq 1$ and  the case when all cognitive users have the same rate, i.e., $R_k=R$ for $k\geq 1$. By symmetry  the parameters $\lambda_k$, $\beta_k$ and $\gamma_k$ should be the same for all $k\geq 1$.  In this symmetric setup,  one simple observation can be made regarding the optimal number of integer sums $L$ and the  coefficient matrix $\ve A$.
\begin{lemma}
For the symmetric many-to-one cognitive interference channel, we need to decode at most two integer sums,  $L\leq 2$.  Furthermore, the optimal coefficient matrix is one of the following two matrices:
\begin{IEEEeqnarray}{rCl}
\ve A_1&=&\begin{pmatrix}
1 &0 &\ldots &0
\end{pmatrix}
\end{IEEEeqnarray}
or
\begin{IEEEeqnarray}{rCl}
\ve A_2&=&\begin{pmatrix}
c_0 &c &\ldots &c\\
0 &1 &\ldots &1
\end{pmatrix}
\end{IEEEeqnarray}
for some integer $c_0$ and nonzero integer $c$.
\label{lemma:optimal_coef}
\end{lemma}
\begin{proof}
For given $\underline\lambda$, $\underline\beta$ and $\underline{\gamma}$, to maximize the rate $R_k$ with respect to $\ve A$ is the same as to minimize the equivalent noise variance $N_0(\ell)$ in (\ref{eq:N_0_l}). We write out $N_0(1)$ for decoding the first equation ($\ell=1$) with $\beta_k=\beta$, $\lambda_k=\lambda$ and $\gamma_k=\gamma$ for all $k\geq 1$:
\begin{IEEEeqnarray*}{rCl}
N_0(1)&=&\alpha_1^2+\sum_{k\geq 1}\left(\alpha_1 b-a_k(1)\beta\right)^2\bar\lambda P +(\alpha_1b_0-a_0(1)\beta_0\\
&&-\gamma\sum_{k\geq 1}a_k(1))^2P
\end{IEEEeqnarray*}

The above expression is symmetric on $a_k(1)$ for all $k\geq 1$ hence the minimum is obtained by letting all $a_k(1)$ be the same. It is easy to see that the same argument holds when we induct on $\ell$, i.e., for any $\ell\in[1:L]$, the minimizing $a_k(\ell)$ is the same  for $k\geq 1$.  Clearly $\ve A_1$ and $\ve A_2$ satisfy this property.

To see why we need at most two integer sums: the case with $\ve A_1$ when the primary decoder decodes one sum is trivial; now consider when it decodes two sums with the coefficients matrix $\ve A_2$. First observe that $\ve A_2$ is in the set $\mathcal A(2)$, meaning we can solve for $\ve t_0$. Furthermore, there is no need to decode a third sum with $a_k(3)$ all equal for $k\geq 1$, because any other sums of this form can be constructed by using the two sums we already have. We also mention that the coefficient matrix
\begin{IEEEeqnarray}{rCl}
\ve A_3&=&\begin{pmatrix}
c_0 &c &\ldots &c\\
1 &0 &\ldots &0
\end{pmatrix}
\end{IEEEeqnarray}
is also a valid choice and will give the same result as $\ve A_2$.
\end{proof}

Now we give some numerical results comparing the proposed scheme with the  conventional schemes proposed in Section \ref{sec:ConventionalCoding} for the symmetric cognitive many-to-one channels.

Figure \ref{fig:rate_region} shows the achievable rate region for a symmetric cognitive many-to-one channel. The dashed and dot-dash lines are achievable regions with DPC in Proposition \ref{prop:DPC} and SND at Rx $0$ in Proposition \ref{prop:JointDecoding}, respectively.   The solid line depicts the  rate region using the proposed scheme in Theorem \ref{thm:rate_cog}. Notice the achievable rates based on the simple  conventional schemes in Proposition \ref{prop:DPC}  and \ref{prop:DPC} are not much better than the trivial time sharing scheme in the multi-user scenario,  due to their inherent inefficiencies on interference suppression. On the other hand, the proposed scheme based on structured codes performs interference alignment in the signal level, which gives better interference mitigation ability at the primary receiver. The effect is emphasized more when we study the non-cognitive system in Section \ref{sec:non-cognitive}. The outer bound in Figure \ref{fig:rate_region} is obtained by considering the system as a two-user multiple-antenna broadcast channel whose capacity region is known. A brief description to this outer bound is given in Appendix  \ref{sec:outer_bound}.

\begin{figure}[!htbp]
\centering
   \includegraphics[scale=0.5]{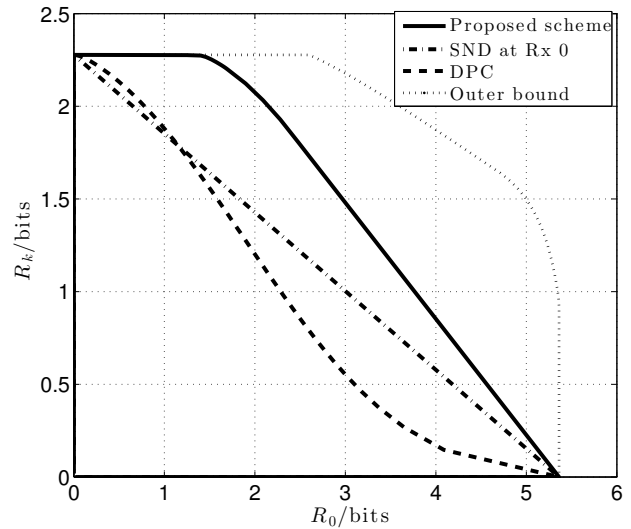}
  \caption{Achievable rate region for a many-to-one symmetric cognitive many-to-one channel  with power $P=10$, channel gain $b_k=4$, $h_k=1.5$ for $k\geq 1$ and $K=3$ cognitive users.The plot compares the different achievable rates for the cognitive many-to-one channel. The horizontal and vertical axis represents the primary rate $R_0$ and cognitive rate $R_k, k\geq 1$, respectively. } 
\label{fig:rate_region}
\end{figure}

It is also instructive to study the system performance as a function of the channel gain $b$.  We consider a symmetric channel with $h$ fixed and varying value of $b$. For different values of $b$,  we  maximize the symmetric rate $R_{sym}:=\min\{R_0,R\}$ where $R=R_k$ for $k\geq 1$ by choosing optimal $\ve A, \underline{\lambda}$ and $\underline{\beta}$, i.e., 
\begin{IEEEeqnarray}{rCl}
  \max_{\substack{\ve A\in \mathcal A(2)\\ \underline{\lambda}, \underline{\beta} }}\min\bigg\{&& \min_{\ell\in\mathcal L_0}r_0(\ve a_{\ell|1:\ell-1}), \min_{\ell\in\mathcal L_k}r_k(\ve a_{\ell|1:\ell-1}), \nonumber\\
  && \max_{\nu_k\in\mathbb R}\frac{1}{2}\log^+\frac{\sigma_k^2}{N_k(\gamma_k)}\bigg\}
\end{IEEEeqnarray}
where the first term is the rate of the primary user and the minimum of the second and the third term is the rate of  cognitive users. Notice  $\lambda_k, \beta_k, r_k(\ve a_{\ell|1:\ell-1})$ are the same for all $k\geq 1$ in this symmetric setup.  Figure \ref{fig:SymRate} shows the maximum symmetric rate of different schemes with increasing $b$.

\begin{figure}[!htbp]
   \centering
   \includegraphics[scale=0.6]{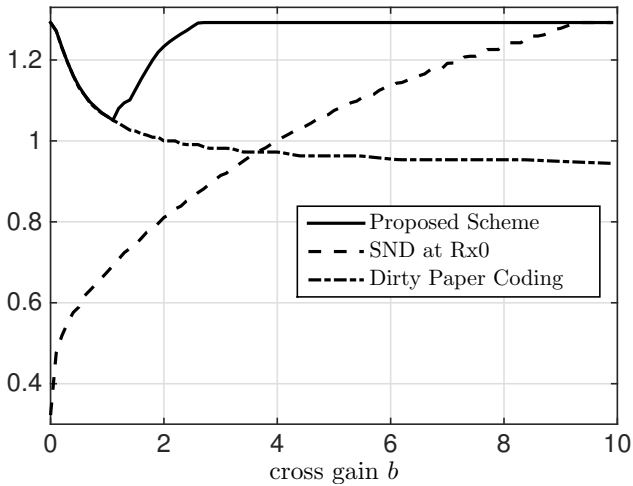}
   \caption{The maximum symmetric rates $R_{sym}$ of different schemes for a many-to-one cognitive interference network with power $P=5$ and $K=3$ cognitive users where $R_k=R$ for $k\geq 1$. We set $h=1$ and vary the cross channel gain $b$ in the interval $[0: 10]$. Notice the maximum symmetric rate is upper bounded by $\frac{1}{2}\log(1+h^2P)$.  We see the proposed scheme performs better than the other two schemes in general. When the interference becomes larger, the proposed scheme quickly attains the maximum symmetric rate.  The joint decoding method approaches the maximum symmetric rate much slower, since it requires the cross channel gain to be sufficiently large such that the primary decoder can (nonuniquely) decode all the messages of the cognitive users. The dirty paper coding approach cannot attain the maximum symmetric rate since the primary decoder treats interference as noise. }
   \label{fig:SymRate}
 \end{figure}

\section{Non-cognitive many-to-one channels}\label{sec:non-cognitive}
As an interesting special case of the cognitive many-to-one channel, in this section we will study the non-cognitive many-to-one channels  where  user $1,\ldots, K$ do not have access to the message $W_0$ of User $0$.  The many-to-one interference channel has also been studied, for example, in \cite{Bresler_etal_2010}, where several constant-gap results are obtained. Using the coding scheme introduced here, we are able to give some refined result to this channel in some special cases.

It is straightforward to extend the coding scheme of  the cognitive channel to the non-cognitive channel by letting users $1,\ldots K$ not split the power for the message $W_0$ but to transmit their own messages only. The achievable rates are the same as in Theorem \ref{thm:rate_cog} by setting all power splitting parameters $\lambda_k$ to be zero and $\gamma_k$ to be zero because $\ve x_0$ will not be interference to cognitive users.  Although it is a straightforward exercise to write out the achievable rates, we still state the result formally here.

\begin{theorem}
For any given positive numbers $\underline\beta$ and coefficient matrix $\ve A\in\mathcal A(L)$ in (\ref{eq:valid_coefficients}) with $L\in[1:K+1]$, define $\mathcal L_k:=\{\ell\in[1:L]|a_k(\ell)\neq 0\}$.  If $r_k(\ve a_{\ell|1:\ell-1},\underline{\lambda},\underline{\beta},\underline{\gamma})>0$ for all $\ell\in\mathcal L_k$, $k\in[0:K]$,  then the following rate is achievable for the many-to-one  interference channel
\begin{IEEEeqnarray}{rCl}
R_0&\leq& \min_{\ell\in\mathcal L_0}\tilde r_0(\ve a_{\ell|1:\ell-1},\underline{\beta})\IEEEyessubnumber \label{eq:R_0_noncog}\\
R_k&\leq &\min\bigg\{\frac{1}{2}\log\left(1+h_k^2P\right),\min_{\ell\in\mathcal L_k}\tilde r_k(\ve a_{\ell|1:\ell-1},\underline{\beta})\bigg\} \IEEEyessubnumber\label{eq:R_k_noncog}
\end{IEEEeqnarray}
for $k\in[1:K]$ with 
\begin{align}
\tilde r_k(\ve a_{\ell|1:\ell-1},\underline{\beta}):=\max_{\alpha_1,\ldots,\alpha_\ell\in\mathbb R}\frac{1}{2}\log^{+}\left(\frac{\beta_k^2 P}{\tilde N_0(\ell)}\right)
 \label{eq:Rk_SIC_noncog}
\end{align}
where $\tilde N_0(\ell)$ is defined as
\begin{IEEEeqnarray}{rCl}
\tilde N_0(\ell)&:=&\alpha_\ell^2+\sum_{k\geq 1}\left(\alpha_\ell b_k-a_k(\ell)\beta_k-\sum_{j=1}^{\ell-1}\alpha_ja_k(j)\beta_k\right)^2 P \nonumber\\
&&+\left(\alpha_\ell-a_0(\ell)\beta_0-\sum_{j=1}^{\ell-1}\alpha_ja_0(j)\beta_0\right)^2P
.\label{eq:N_0_l_noncog}
\end{IEEEeqnarray}
\label{thm:rate_noncog}
\end{theorem}

\begin{proof}
The proof of this result is almost the same as the proof of Theorem \ref{thm:rate_cog} in Section \ref{sec:Proof_Thm_cog}. The only change in this proof is that the user $1,\ldots, K$ do not split the power to transmit for the primary user and all $\gamma_k$ are set to be zero since $\ve x_0$ will not act as interference to cognitive receivers.   We will use  lattice codes  described in Section \ref{subsec:LatticeCodes} but  adjust the code construction.  Given  positive numbers $\underline\beta$ and a simultaneously good fine lattice $\Lambda$, we choose $K+1$ simultaneously good  lattices such that $\Lambda_k^s\subseteq \Lambda_k$ with second moments $\sigma^2(\Lambda_k^s)=\beta_k^2 P$ for all $k\in[0:K]$.

Each user forms the  transmitted signal as 
\begin{IEEEeqnarray}{rCl}
\ve x_k&=&\left[\frac{\ve t_k}{\beta_k}+\ve d_k\right]\mode\Lambda_k^s/\beta_k,\quad k\in[0:K]
\end{IEEEeqnarray}
The analysis of the decoding procedure at all receivers is the same as in Section \ref{sec:Proof_Thm_cog}. User $0$ decodes integer sums to recover $\ve t_0$ and other users decode their message $\ve t_k$ directly from the channel output using lattice decoding. In fact, the expression $\tilde r_k(\ve a_{\ell|1:\ell-1},\underline{\beta})$ in (\ref{eq:Rk_SIC_noncog}) is the same as $r_k(\ve a_{\ell|1:\ell-1},\underline{\lambda},\underline{\beta},\underline{\gamma})$ in (\ref{eq:Rk_SIC}) by letting $\lambda_k=\gamma_k=0$ in the later expression. Furthermore we have
\begin{align}
\max_{\nu_k\in\mathbb R}\frac{1}{2}\log\frac{\sigma_k^2}{N_k(\gamma_k=0)}=\frac{1}{2}\log(1+h_k^2P)
\end{align}  for any choice of $\beta_k, k\geq 1$.
\end{proof}

For a simple symmetric example, we compare the achievable rate region of the cognitive many-to-one channel  (Theorem \ref{thm:rate_cog}) with the achievable rate region of the non-cognitive many-to-one channel (Theorem \ref{thm:rate_noncog}) in Figure \ref{fig:cog_vs_noncog}. The parameters are the same  for both channel. This shows the usefulness of the cognitive messages in the system.

\begin{figure}[!htbp]
\centering
   \includegraphics[scale=0.5]{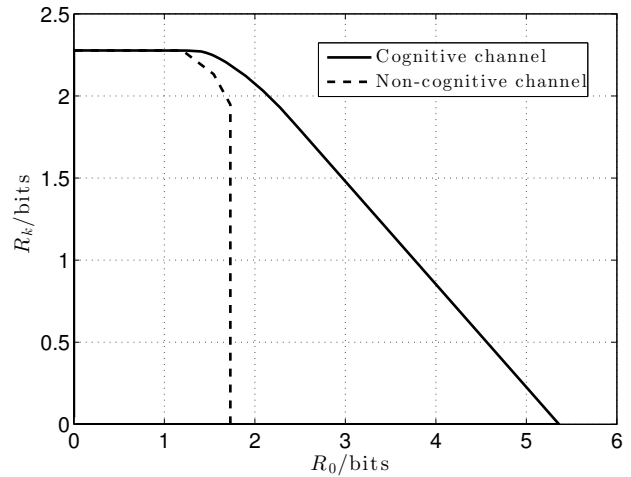}
  \caption{A many-to-one symmetric  interference channel with power $P=10$, channel gain $b_k=4$, $h_k=1.5$ for $k\geq 1$ and $K=3$ cognitive users. This plot compares the different achievable rate regions for the cognitive and non-cognitive channel.  The horizontal and vertical axis represents the primary rate $R_0$ and cognitive rate $R_k, k\geq 1$, respectively. The rate region for the cognitive channel given by Theorem \ref{thm:rate_cog} is plotted in solid line. The dashed line gives the achievable rate region  in Theorem \ref{thm:rate_noncog} for the non-cognitive many-to-one channel. } 
\label{fig:cog_vs_noncog}
\end{figure}

\subsection{Capacity Results for  Non-cognitive Symmetric  Channels}\label{sec:noncog}
Now we consider a symmetric non-cognitive many-to-one channel where $b_k=b$ and $h_k=h$ for $k\geq 1$. In \cite{Bresler_etal_2010},  an approximate capacity result is established within a gap of $(3K+3)(1+\log(K+1))$ bits per user for \textit{any} channel gain. In this section we will give refined results for the symmetric many-to-one channel. The reason we restrict ourselves to the symmetric case is that, for general channel gains the optimization problem involving the coefficient matrix $\ve A$ is analytically intractable as  discussed in Section \ref{sec:optimal_A}, hence it is also difficult to give explicit expressions for achievable rates. But for the symmetric many-to-one channel we are able to  give a constant gap result as well as a capacity result when the interference is strong. First notice that the optimal form of the coefficient matrix for the cognitive symmetric channel given in Lemma \ref{lemma:optimal_coef}  also applies in this non-cognitive symmetric setting.

\begin{theorem}
Consider a symmetric (non-cognitive) many-to-one interference channel  with $K+1$ users. If $|b|\geq|h|\ceil[\Big]{\sqrt{P}}$, then each user is less than $0.5$ bit from the capacity for any  number of users. Furthermore, if $|b|\geq \sqrt{\frac{(1+P)(1+h^2P)}{P}}$, each user can achieve the capacity, i.e., $R_0=\frac{1}{2}\log(1+P)$ and $R_k=\frac{1}{2}\log(1+h^2P)$ for all $k\geq 1$. 
\label{thm:constant_gap}
\end{theorem}

\begin{proof}
For the symmetric non-cognitive many-to-one channel, we have the following  trivial capacity bound
\begin{IEEEeqnarray}{rCl}
R_0&\leq& \frac{1}{2}\log(1+P)\\
R_k&\leq& \frac{1}{2}\log(1+h^2P).
\end{IEEEeqnarray}

To show the constant gap result, we choose the coefficients matrix of the two sums to be
\begin{IEEEeqnarray}{rCl}
\ve A&=&\begin{pmatrix}
1 &c &\ldots &c\\
0 &1 &\ldots &1
\end{pmatrix}
\label{eq:matrix_1}
\end{IEEEeqnarray}
for some nonzero integer $c$. Furthermore we choose $\beta_0=1$ and $\beta_k=b/c$ for all $k\geq 1$. In Appendix \ref{sec:Appdx_derivation} we use  Theorem \ref{thm:rate_noncog} to show  the following rates are achievable:
\begin{IEEEeqnarray*}{rCl}
R_0&=& \frac{1}{2}\log^+P\\
R_k&=&\min\bigg\{\frac{1}{2}\log^+\frac{b^2P}{c^2}, \frac{1}{2}\log^+b^2, \frac{1}{2}\log(1+h^2P)\bigg\}.
\end{IEEEeqnarray*}
If $|b|\geq |h|\ceil[\Big]{\sqrt{P}}$, choosing $c=\ceil[\Big]{\sqrt{P}}$  will ensure $R_k\geq \frac{1}{2}\log^+ h^2P$.

Notice that for $P\leq 1$, then $\frac{1}{2}\log(1+P)\leq 0.5$ hence the claim is vacuously true.  For $P\geq 1$, we have 
\begin{IEEEeqnarray*}{rCl}
\frac{1}{2}\log(1+P)-R_0\leq  \frac{1}{2}\log\frac{1+P}{P}\leq\frac{1}{2}\log 2= 0.5 \mbox{ bit}
\end{IEEEeqnarray*}
With the same argument we have 
\begin{IEEEeqnarray}{rCl}
\frac{1}{2}\log(1+h^2P)-R_k\leq 0.5 \mbox{ bit}
\end{IEEEeqnarray}

To show the capacity result, we set $\beta_0=1$ and $\beta_k=\beta$ for all $k\geq 1$. The receiver $0$ decodes two sums with the coefficients matrix
\begin{align}
\ve A=\begin{pmatrix}
0 &1 &\ldots &1\\
1 &0 &\ldots &0
\end{pmatrix}
\label{eq:matrx_2}.
\end{align}
The achievable rates using Theorem \ref{thm:rate_noncog} is shown in Appendix \ref{sec:Appdx_derivation}  to be
\begin{IEEEeqnarray}{rCl}
R_0&=&\frac{1}{2}\log(1+P)\\
R_k&=& \min\left\{\frac{1}{2}\log\left(\frac{Pb^2}{1+P}\right),\frac{1}{2}\log(1+h^2P)\right\}.
\end{IEEEeqnarray}
The inequality
\begin{IEEEeqnarray}{rCl}
\frac{Pb^2}{1+P}&\geq& 1+h^2P
\end{IEEEeqnarray}
is satisfied if it holds that
\begin{IEEEeqnarray}{rCl}
b^2\geq \frac{(1+P)(1+h^2P)}{P}.
\end{IEEEeqnarray}
This completes the proof.
\end{proof}

Comparing to the constant gap result in \cite{Bresler_etal_2010}, our result only concerns   a special class of many-to-one channel, but gives a gap which does not depend on the number of users $K$.  
We also point out that in \cite{sridharan_capacity_2008}, a $K$-user symmetric interference channel is studied where it was shown that if the cross channel gain $h$ satisfies $|h|\geq \sqrt{\frac{(1+P)^2}{P}}$, then every user achieves the capacity $\frac{1}{2}\log(1+P)$. This result is very similar to our result obtained here and is actually obtained using the same coding technique.

\appendices

\section{Derivations in the proof of Theorem \ref{thm:rate_cog}}
\label{app:SIC}
We give the proof for the claim made in Section \ref{sec:Proof_Thm_cog} that we could form the equivalent channel
\begin{IEEEeqnarray*}{rCl}
\tilde{\ve y}_0^{(\ell)}=\tilde{\ve z}_0(\ell)+\sum_{k\geq 0}a_k(\ell)\tilde{\ve t}_k
\end{IEEEeqnarray*}
with $\tilde{\ve z}_0(\ell)$ defined in (\ref{eq:z_0_l_tilde})
when the primary decoder decodes the $\ell$-th integer sum $\sum_{k\geq 0}a_k(\ell)\tilde{\ve t}_k$ for $\ell\geq 2$.

We first show the base case for $\ell=2$. Since $\sum_{k\geq 0}a_k(1)\tilde{\ve t}_k$ is decoded, the equivalent noise $\tilde{\ve z}_0(1)$ in Eqn. (\ref{eq:z_0_1_tilde}) can be inferred from $\tilde{\ve y}_0$.  Given $\alpha_{20},\alpha_{21}$ we form the following with $\ve y_0$ in (\ref{eq:y_0}) and $\tilde{\ve z}_0(1)$
\begin{IEEEeqnarray*}{rCl}
\tilde{\ve y}_0^{(2)}&:=&\alpha_{20}\ve y_0+\alpha_{21}\tilde{\ve z}_0(1)\\
&=&(\alpha_{20}+\alpha_{21}\alpha_1)\ve z_0\\
&&+\sum_{k\geq 1}( (\alpha_{20}+\alpha_{20}\alpha_1)b_k-\alpha_{21}a_k(1)\beta_k)\hat{\ve x}_k\\
&&+((\alpha_{20}+\alpha_{21}\alpha_1)b_0-\alpha_{21}a_0(1)\beta_0-
\alpha_{21}g(1))\ve x_0\\
&=&\alpha_2'\ve z_0+\sum_{k\geq 1}(\alpha_2'b_k-\alpha_1'a_k(1)\beta_k)\hat{\ve x}_k\\
&&+(\alpha_2' b_0-\alpha_1'a_0(1)\beta_0-\alpha_1'g(1))\ve x_0
\end{IEEEeqnarray*}
by defining $\alpha_1':=\alpha_{21}$ and $\alpha_2':=\alpha_{20}+\alpha_{21}\alpha_1$. Now following the same step for deriving $\tilde{\ve y}_0^{(1)}$ in (\ref{eq:y0_tilde}), we can rewrite $\tilde{\ve y}_0^{(2)}$ as
\begin{IEEEeqnarray}{rCl}
\tilde{\ve y}_0^{(2)}=\sum_{k\geq 0}a_k(2)\tilde{\ve t}_k+\tilde{\ve z}_0(2)
\end{IEEEeqnarray}
with
\begin{IEEEeqnarray*}{rCl}
\tilde{\ve z}_0(2)&:=&\alpha_2'\ve z_0+\sum_{k\geq 1}(\alpha_2'b_k-a_k(2)\beta_k-\alpha_1'a_k(1)\beta_k)\hat{\ve x}_k\\
&&+(\alpha_2' b_0-a_0(2)\beta_0-\alpha_1'a_0(1)\beta_0-g(2))\ve x_0
\end{IEEEeqnarray*}
This establishes the base case by identifying $\alpha_i'=\alpha_i$ for $i=1, 2$.  

Now assume the expression (\ref{eq:z_0_l_tilde}) is true for $\ell-1$ ($\ell\geq 3$) and we have inferred $\tilde{\ve z}_0(m)$  from $\tilde{\ve y}_0^{(m)}$ using the decoded sum $\sum_{k\geq 0}a_k(m)\tilde{\ve t}_k$ for all $m\leq 1,\ldots,\ell-1$,  we will form $\tilde{\ve y}_0^{(\ell)}$ with $\ell$ numbers $\alpha_{\ell 0},\ldots,\alpha_{\ell \ell-1}$ as 
\begin{IEEEeqnarray*}{rCl}
\tilde{\ve y}_0^{(\ell)}&:=&\alpha_{\ell 0}\ve y_0+\sum_{m=1}^{\ell-1}\alpha_{\ell m}\tilde{\ve z}_0(m)\\
&=&\alpha'_\ell\ve z_0+\sum_{k\geq 1}\left(\alpha_\ell' b_k-\beta_k C_{\ell-1}(k)\right)\hat{\ve x}_k\\
&&+\left(\alpha'_\ell b_0-\beta_0C_{\ell-1}(0)-\sum_{m=1}^{\ell-1}\alpha_{\ell m}g(m)\right)\ve x_0
\end{IEEEeqnarray*}
with 
\begin{align}
\alpha'_\ell&:=\alpha_{\ell 0}+\sum_{m=1}^{\ell-1}\alpha_{\ell m}\alpha_m\\
C_{\ell-1}(k)&:=\sum_{m=1}^{\ell-1}\alpha_{\ell m}\left(a_k(m)+\sum_{j=1}^{m-1}\alpha_ja_k(j)\right).
\end{align}
Algebraic manipulations allow us to rewrite $C_{\ell-1}(k)$ as
\begin{align}
C_{\ell-1}(k)&=\sum_{m=1}^{\ell-1}\left(\alpha_{\ell m}+\alpha_m\sum_{j=m+1}^{\ell-1}\alpha_{\ell j}\right)a_k(m)\\
&=\sum_{m=1}^{\ell-1}\alpha'_ma_k(m)
\end{align}
by defining $\alpha'_m:=\alpha_{\ell m}+\alpha_m\sum_{j=m+1}^{\ell-1}\alpha_{\ell j}$ for $m=1,\ldots,\ell-1$. Substituting the above into $\tilde{\ve y}_0^{(\ell)}$  we get
\begin{IEEEeqnarray*}{rCl}
\tilde{\ve y}_0^{(\ell)}&=&\alpha'_\ell \ve z_0+\sum_{k\geq 1}\left(\alpha_\ell' b_k-\beta_k \sum_{m=1}^{\ell-1}\alpha'_ma_k(m)\right)\hat{\ve x}_k\\
&&+\left(\alpha'_\ell b_0-\beta_0\sum_{m=1}^{\ell-1}\alpha'_ma_0(m)-\sum_{m=1}^{\ell-1}\alpha_{\ell m}g(m)\right)\ve x_0.
\end{IEEEeqnarray*}
Together with the definition of $g(m)$ in (\ref{eq:g_l}) and some algebra we can show
\begin{align}
\sum_{m=1}^{\ell-1}a_{\ell m}g(m)&=\sum_{k=1}^{K}\gamma_kC_{\ell-1}(k)\\
&=\sum_{k=1}^K\left(\sum_{m=1}^{\ell-1}\alpha_m'a_k(m)\right) \gamma_k.
\end{align}
Finally using the same steps for deriving $\tilde{\ve y}_0^{(1)}$ in (\ref{eq:y0_tilde}) and identifying $\alpha'_m=\alpha_m$ for $m=1,\ldots,\ell$, it is easy to see that we have
\begin{IEEEeqnarray}{rCl}
\tilde{\ve y}_0^{(\ell)}&=&\sum_{k\geq 0}a_k(\ell)\tilde{\ve t}_k+\tilde{\ve z}_0(\ell)
\end{IEEEeqnarray}
with $\tilde{\ve z}_0(\ell)$ claimed in (\ref{eq:z_0_l_tilde}).

\section{Proof of Proposition \ref{prop:optimization}}\label{sec:proof_prop}
For any given  set of parameters $\{\alpha_j, j\in[1:\ell]\}$ in the expression $N_0(\ell)$ in (\ref{eq:N_0_l_compact}) , we can always find another set of parameters $\{\alpha_j', j\in[1:\ell]\}$ and a set of vectors $\{\ve u_j, j\in[1:\ell]\}$, such that
\begin{align}
\alpha_\ell\ve h+\sum_{j=1}^{\ell-1}\alpha_j\tilde{\ve a}_j=\sum_{j=1}^{\ell}\alpha_j'\ve u_j
\end{align}
as long as the two sets of vectors, $\{\ve h, \tilde{\ve a}_j, j\in[1:\ell-1]\}$ and $\{\ve u_j, j\in[1:\ell]\}$ span the same subspace. If we choose an appropriate set of basis vectors $\{\ve u_j\}$, the minimization problem of $N_0(\ell)$ can be equivalently formulated with the set $\{\ve u_j\}$ and new parameters $\{\alpha_j'\}$ where the optimal $\{\alpha_j'\}$ have simple solutions. 
Notice that $\{\ve u_j,j\in[1:\ell]\}$ in Eqn. (\ref{eq:u_j}) are obtained by performing the Gram-Schmidt procedure on the set  $\{\ve h, \tilde{\ve a}_j, j\in[1:\ell-1]\}$. Hence the set $\{\ve u_j, j\in[1:\ell]\}$  contains orthogonal vectors and spans the same subspace as the set  $\{\ve h, \tilde{\ve a}_j, j\in[1:\ell-1]\}$ does. For any $\ell\geq 1$, the expression $N_0(\ell)$ in (\ref{eq:N_0_l_compact}) can be equivalently rewritten as
\begin{align}
N_0(\ell)=\alpha_{\ell}'^2+\norm{\sum_{j=1}^{\ell}\alpha_j'\ve u_j-\tilde {\ve a}_{\ell}}^2P
\label{eq:N_0_l_newbasis}
\end{align}
with  $\{\ve u_j\}$  defined above and some $\{\alpha_j'\}$.   Due to the orthogonality of vectors $\{\ve u_j\}$, we have the following simple optimal solutions for $\{\alpha_j'^{*}\}$ which minimize $N_0(\ell)$:
\begin{IEEEeqnarray}{rCl}
\alpha_j'^{*}&=&\frac{\tilde{\ve a}_\ell^T\ve u_j}{\norm{\ve u_j}^2},\quad j\in[1:\ell-1]\\
\alpha_\ell'^{*}&=&\frac{P\tilde{\ve a}_{\ell}^T\ve u_\ell}{P\norm{\ve u_{\ell}}^2+1}.
\end{IEEEeqnarray}
Substituting them back to $N_0(\ell)$ in (\ref{eq:N_0_l_newbasis}) we have
\begin{IEEEeqnarray*}{rCl}
N_0(\ell)&=&P\norm{\tilde{\ve a}_\ell}^2-\sum_{j=1}^{\ell-1}\frac{(\tilde{\ve a}_\ell^T\ve u_j)^2P}{\norm{\ve u_j}^2}-\frac{P^2(\ve u_\ell^T\tilde{\ve a}_\ell)^2}{1+P\norm{\ve u_\ell}^2}\\
&=&P\tilde{\ve a}_\ell^T\left(\ve I-\sum_{i=1}^{\ell-1}\frac{\ve u_j\ve u_j^T}{\norm{\ve u_j}^2}-\frac{(\ve u_\ell\ve u_\ell^T)P}{1+P\norm{\ve u_\ell}^2}\right)\tilde{\ve a}_\ell\\
&=&P\ve a(\ell)^T\ve B_{\ell}\ve a(\ell)
\end{IEEEeqnarray*}
with $\ve B_\ell$ given in (\ref{eq:matrix_B}). As we discussed before, maximizing $r_k(\ve a_{\ell|1:\ell-1})$ is equivalent to minimizing $N_0(\ell)$ and the optimal coefficients $\ve a(\ell), \ell\in[1:L]$ are the same for all users. This proves the claim.

\section{An outer bound on the capacity region} \label{sec:outer_bound}
In this section we give a simple outer bound on the capacity region of the cognitive many-to-one channel, which is used for the numerical evaluation in Figure \ref{fig:rate_region}, Section \ref{sec:Symmetric}. Notice that if we allow all  transmitters $k=0,\ldots,K$ to cooperate, and allow the cognitive receivers $k=1,\ldots, K$ to cooperate, then the system can be seen as a $2$-user broadcast channel  where the transmitter has $K+1$ antennas.  The two users are the primary receiver and the aggregation of all cognitive receivers with $K$ antennas.  Obviously the capacity region of this resulting $2$-user MIMO broadcast channel will be a valid outer bound on the capacity region of the cognitive many-to-one channel. The capacity region $\mathcal C_{BC}$ of the broadcast channel is given by (see \cite[Ch. 9]{Elgamal_Kim_2011} for example)
\begin{IEEEeqnarray}{rCl}
\mathcal C_{BC}=\mathcal R_1\bigcup\mathcal R_2
\end{IEEEeqnarray}
where $\mathcal R_1$ is defined as
\begin{IEEEeqnarray}{rCl}
R_1&\leq& \frac{1}{2}\log\frac{|\ve H_1(\ve K_1+\ve K_2)\ve H_1^T+\ve I|}{|\ve H_1\ve K_2\ve H_2^T+\ve I|}\\
R_2&\leq& \frac{1}{2}\log|\ve H_2\ve K_2\ve G_2^T+\ve I|
\end{IEEEeqnarray}
and $\mathcal R_2$ defined similarly with all subscripts $1$ and $2$ in $\mathcal R_1$ swapped. The channel matrices $\ve H_1\in\mathbb R^{1\times (K+1)}$ and $\ve H_2\in\mathbb R^{K\times (K+1)}$ are defined as
\begin{IEEEeqnarray}{rCl}
\ve H_1&=&\begin{bmatrix}
1 &b_1 &\ldots &b_K
\end{bmatrix}\\
\ve H_2&=&
\begin{bmatrix}
0 &h_1 &0 &\ldots &0\\
0 &0 &h_2 &\ldots &0\\
\vdots &\vdots &\vdots &\ddots &\vdots \\
0 &0 &0 &\ldots &h_K
\end{bmatrix}
\end{IEEEeqnarray}
where $\ve H_1$ denotes the channel from the aggregated transmitters to the primary receiver and $\ve H_2$ denotes the channel to all cognitive receivers. The variables $\ve K_1,\ve K_2\in\mathbb R^{(K+1)\times (K+1)}$ should satisfy the condition
\begin{IEEEeqnarray}{rCl}
 \mbox{tr}(\ve K_1+\ve K_2)\leq (K+1)P
\end{IEEEeqnarray}
which represents the power constraint for the corresponding broadcast channel\footnote{Since each transmitter has its individual power constraint, we could give a slightly tighter outer bound by imposing a per-antenna power constraint. Namely the matrices $\ve K_1, \ve K_2$ should satisfy $(\ve K_1+\ve K_2)_{ii}\leq P$ for $i\in[1:K+1]$ where $(\ve X)_{ii}$ denotes the $(i,i)$ entry of  matrix $\ve X$. However this is not the focus of this paper and we will not pursue it here.}. As explained in \cite[Ch. 9]{Elgamal_Kim_2011}, the problem of finding the region $\mathcal C_{BC}$ can be rewritten as convex optimization problems which are readily solvable using standard convex optimization tools.

\section{Derivations in the proof of Theorem \ref{thm:constant_gap}}\label{sec:Appdx_derivation}
We give detailed derivations of the achievable rates in Theorem \ref{thm:constant_gap} with two chosen coefficient matrices. 

When the primary user decodes the first equation ($\ell=1$) in a symmetric channel, the expression (\ref{eq:N_0_l_noncog})  for  the variance of the equivalent noise simplifies to (denoting $\beta_k=\beta$ for $k\geq 1$)
\begin{align}
\tilde N_0(1)=\bar\alpha_1^2+K(\bar\alpha_1b-a_k(1)\beta)^2P+(\bar\alpha_1-a_0(1)\beta_0)^2P.
\label{eq:N_0_1_example}
\end{align}
For decoding the second integer sum, the variance of the equivalent noise (\ref{eq:N_0_l_noncog}) is given as
\begin{IEEEeqnarray}{rCl}
\tilde N_0(2)&=&\alpha_2^2+K(\alpha_2b-a_k(2)\beta-\alpha_1a_k(1)\beta)^2P\nonumber\\
&&+(\alpha_2-a_0(2)\beta_0-\alpha_1a_0(1)\beta_0)^2P.
\label{eq:N_0_2_example}
\end{IEEEeqnarray}

We first evaluate the achievable rate for the coefficient matrix in (\ref{eq:matrix_1}). We choose $\beta_0=1$ and  $\beta=b/c$. Using Theorem \ref{thm:rate_noncog},  substituting $\ve a(1)=[1,c,\ldots,c]$ and the optimal $\bar\alpha_1^*=1-\frac{1}{P(Kb^2+1)}$ into (\ref{eq:N_0_1_example}) will give us a rate constraint on $R_0$
\begin{IEEEeqnarray*}{rCl}
\tilde r_0(\ve a_1,\underline{\beta})&=&\frac{1}{2}\log^+\left(\frac{1}{1+Kb^2}+P\right)>\frac{1}{2}\log^+ P\\
\tilde r_k(\ve a_1,\underline{\beta})&=&\frac{1}{2}\log^+\left(\frac{b^2P(Kb^2P+P+1)}{c^2(Kb^2P+P)}\right)>\frac{1}{2}\log^+\frac{b^2P}{c^2}.
\end{IEEEeqnarray*}
Notice here we have replaced the achievable rates with smaller values to make the result simple. We will do the same in the following derivation.

For decoding the second sum with coefficients $\ve a(2)=[0,1,\ldots,1]$, we use  Theorem \ref{thm:rate_noncog} and (\ref{eq:N_0_2_example})  to obtain rate constraints for $R_k$
\begin{IEEEeqnarray}{rCl}
\tilde r_k(\ve a_{2|1},\underline{\beta})=\frac{1}{2}\log^+\left(b^2+\frac{1}{K}\right)> \frac{1}{2}\log^+b^2
\end{IEEEeqnarray}
with the optimal $\alpha_1^*=\frac{-b^2K}{c(Kb^2+1)}$ and $\alpha_2^*=0$.  Notice that  $\ve a_0(1)=0$  hence decoding  this sum will not impose any rate constraint on $R_0$. Therefore we omit the expression $\tilde r_0(\ve a_{2|1},\underline{\beta})$. Combining the results above with Theorem \ref{thm:rate_noncog} we get the claimed rates in the proof of Theorem \ref{thm:constant_gap}.

Now we evaluate the achievable rate for the  coefficient matrix in (\ref{eq:matrx_2}). We substitute $\beta_0=1$, $\beta_k=\beta$ for any $\beta$ and $\ve a(1)=[0,1,\ldots,1]$ in (\ref{eq:N_0_1_example}) with the optimal $\bar\alpha_1^*=\frac{Kb\beta p}{Kb^2P+P+1}$. Notice again $R_0$ is not constrained by decoding this sum hence we only have the constraint on $R_k$ as 
\begin{IEEEeqnarray*}{rCl}
\tilde r_k(\ve a_1,\underline{\beta})&=&\frac{1}{2}\log^+\left(\frac{1}{K}+\frac{P}{1+P}b^2\right)>\frac{1}{2}\log^+\frac{Pb^2}{1+P}.
\end{IEEEeqnarray*}
For the second decoding, using $\ve a(2)=[1,0,\ldots,0]$ in (\ref{eq:N_0_2_example}) gives  
\begin{IEEEeqnarray}{rCl}
\tilde r_0(\ve a_{2|1},\underline{\beta})=\frac{1}{2}\log\left(1+P\right)
\end{IEEEeqnarray}
with the optimal scaling factors $\alpha_1^*=\frac{bP}{\beta(P+1)}$ and $\alpha_2^*=\frac{P}{P+1}$. Combining the achievable rates above with  Theorem \ref{thm:rate_noncog}  gives the claimed result.


\section*{Acknowledgment}

The authors would like to thank the anonymous reviewers
for their valuable comments and suggestions to improve the
quality of the paper.

\ifCLASSOPTIONcaptionsoff
  \newpage
\fi

\bibliographystyle{IEEEtran}
\bibliography{IEEEabrv,ZhuGastpar}


%

\begin{IEEEbiographynophoto}{Jingge Zhu}
is a Ph.D. student in the  School of Computer and Communication Sciences, Ecole Polytechnique F{\'e}d{\'e}rale de Lausanne (EPFL), Lausanne, Switzerland. He received the B.S. degree and M.S. degree in electrical engineering from Shanghai Jiao Tong University, Shanghai, China, in 2008 and 2011, respectively. He also received the Dipl.-Ing. degree in technische Informatik from Technische Universit{\"a}t Berlin, Berlin, Germany in 2011.  His research interests include  information theory with applications in communication systems.

Mr. Zhu is the recipient of the IEEE Heinrich Hertz Award for Best
Communications Letters in 2013.
\end{IEEEbiographynophoto}

\begin{IEEEbiographynophoto}{Michael Gastpar}
received the Dipl. El.-Ing. degree from ETH Z\"urich, in 1997, the M.S. degree from the University of Illinois at Urbana-Champaign, Urbana, IL, in 1999, and the
Doctorat \`es Science degree from Ecole Polytechnique F\'ed\'erale (EPFL), Lausanne, Switzerland, in 2002, all in electrical engineering. He was also a student in engineering and philosophy at the Universities of Edinburgh and Lausanne.

He is a Professor in the School of Computer and Communication
Sciences, Ecole Polytechnique F\'ed\'erale (EPFL), Lausanne, Switzerland.
He held tenured professor positions at the University of California, Berkeley,
and at Delft University of Technology, The Netherlands.
He was a Researcher with the Mathematics of Communications Department,
Bell Labs, Lucent Technologies, Murray Hill, NJ.
His research interests are
in network information theory and related coding and signal processing techniques,
with applications to sensor networks and neuroscience.

Dr. Gastpar received the 2002 EPFL Best Thesis Award, the NSF CAREER Award
in 2004, the Okawa Foundation Research Grant in 2008, the ERC Starting Grant in 2010,
and the IEEE Communications Society and Information Theory Society Joint Paper Award
in 2013. He was an Information Theory Society Distinguished Lecturer (2009-2011), an Associate Editor for Shannon Theory for the IEEE TRANSACTIONS ON INFORMATION
THEORY (2008-2011), and he has served as Technical Program Committee Co-Chair for the
2010 International Symposium on Information Theory, Austin, TX.
\end{IEEEbiographynophoto}





\end{document}